\documentclass[9pt,journal,twocolumn]{IEEEtran}
\usepackage{setspace}
\usepackage{graphics} 
\usepackage{graphicx}
\usepackage{epsfig} 
\usepackage{float}
\usepackage{times} 
\usepackage{amsthm} 
\usepackage{stdmath}
\usepackage{amsmath}  
\usepackage{amsfonts}
\newtheorem{theorem}{Theorem}
\newtheorem{definition}{Definition}
\newtheorem{lemma}{Lemma}

\newtheorem{corollary}{Corollary}
\usepackage{wrapfigure}
\allowdisplaybreaks
\usepackage{graphicx}
\usepackage{caption}
\usepackage{subfigure}
\makeatother

\newcommand{\pfs}{\frac{\partial}{\partial s}}

\newcommand{\igzo}{\int_0^1}
\newcommand{\igzx}{\int_0^x}
\newcommand{\igxo}{\int_x^1}

\newcommand{\igzs}{\int_0^s}
\newcommand{\igso}{\int_s^1}

\newcommand{\wh}{\hat{w}}
\newcommand{\wt}{\tilde{w}}

\newcommand{\lt}{L_2(0,1)}

\newcommand{\mcl}[1]{\mathcal{#1}}
\newcommand{\pft}{\frac{\partial}{\partial t}}
\newcommand{\hlf}{\frac{1}{2}}
\usepackage{tikz, xcolor}
\usetikzlibrary{shapes,arrows}
\usetikzlibrary{positioning}
\usetikzlibrary{calc}
\usetikzlibrary{fit}
\tikzset{
  -|-/.style={
    to path={
      (\tikztostart) -| ($(\tikztostart)!#1!(\tikztotarget)$) |- (\tikztotarget)
      \tikztonodes
    }
  },
  -|-/.default=0.5,
  |-|/.style={
    to path={
      (\tikztostart) |- ($(\tikztostart)!#1!(\tikztotarget)$) -| (\tikztotarget)
      \tikztonodes
    }
  },
  |-|/.default=0.5,
}

\tikzstyle{decision} = [diamond, draw, text width=4.5em,
                        text badly centered, node distance=2cm,
                        inner sep=0pt]
\tikzstyle{block} = [rectangle, draw, text width=10cm,
                     text centered,
                     minimum height=1cm, node distance=3cm]

\tikzstyle{line} = [draw, -latex']
\tikzstyle{smallblock} = [rectangle, draw,
                     text centered,
                     minimum height=1cm, node distance=3cm]
\tikzstyle{cloud} = [draw, circle, node distance=2.5cm, minimum height=2.8em]
\tikzstyle{blank} = [node distance=1cm]
\begin{document}
%
\title{A Convex Approach to Output Feedback Control of Parabolic PDEs Using Sum-of-Squares}
%
%
%

\author{Aditya~Gahlawat
         and~Matthew.~M.~Peet
\thanks{Aditya Gahlawat is with the Department
of Mechanical, Materials and Aerospace Engineering at the Illinois Institute of Technology, Chicago,
IL, 60616 USA and with the Grenoble Image Parole Signal Automatique Lab., Universit\'{e} Joseph Fourier/Centre National de la Recherche Scientifique, St. Martin d'Heres, France e-mail: (agahlawa@hawk.iit.edu).}
\thanks{Matthew. M. Peet is with the School of Engineering of Matter, Transport and Energy at Arizona State University, Tempe, AZ, 85287-6106 USA e-mail: (mpeet@asu.edu).}
}


\maketitle

\begin{abstract}
In this paper we use optimization-based methods to design output-feedback controllers for a class of one-dimensional parabolic partial differential equations. The output may be distributed or point-measurements. The input may be distributed or boundary actuation. We use Lyapunov operators, duality, and the Luenberger observer framework to reformulate the synthesis problem as a convex optimization problem expressed as a set of Linear-Operator-Inequalities (LOIs). We then show how feasibility of these LOIs may be tested using Semidefinite Programming (SDP) and the Sum-of-Squares methodology.
\end{abstract}
\IEEEpeerreviewmaketitle

\section{Introduction}

Parabolic Partial Differential Equations (PDEs) are a simple class of system used to model processes such as diffusion, transport and reaction.  Some examples of systems which have been modelled using Parabolic PDEs include plasma in a tokamak~\cite{witrant2007control}, heat propagation, and spatial dynamics of population in an ecosystem~\cite{murray2002mathematical}. Despite the wide variety of physical phenomena modeled by partial-differential equations, our knowledge of how to control these systems is underdeveloped. While much attention has focused on the use of advanced computing strategies for simulation of partial-differential equations, relatively little work has focused on the development of numerical methods for control of PDEs. This is in particular contrast to the state of the art for linear Ordinary Differential Equations (ODEs), wherein Linear Matrix Inequalities and Convex optimization have been used to resolve a vast array of long-standing problems - e.g. $H_\infty$-optimal output feedback. The goal of this paper, then, is to attempt to extend some of the computational methods for control of linear ODEs to control of linear PDEs.

Differential models incorporating multiple independent variables (e.g. time and space) have been around since the time of Newton. Indeed, many of the models we use today date from this time - e.g. D'Alembert and the wave equation; the Euler-Bernoulli beam; The Euler Equations. Although research into PDEs over the past century has mainly focused on constructing analytic or numerical solutions to these systems, an effort has also been made to define a framework for control. One facet of this research into defining a framework for control of PDEs has been to define a general class of forward-time PDE systems using the label of ``strongly-continuous semigroup''. For such systems, existence and continuity of solutions is guaranteed for bounded feedback operators. See~\cite{curtain1995introduction,bensoussan1995representation,
hale1971functional,lasiecka2000control} for several excellent volumes on this subject.
One of the advantages of a well-defined state-space is the ability to use Lyapunov analysis to prove properties of the state. Indeed, application of Lyapunov theory to infinite-dimensional systems has been studied for some time - See early results in~\cite{krasovskistability,datko1970extending,baker1969lyapunov}.

PDE models of control can vary significantly based on the type of PDE, boundary conditions, measurements, etc. Unlike ODE systems, these differences may dramatically alter the definition of state and other mathematical properties of the solution. For instance, control of PDEs can be classified as either distributed input or boundary/point input. For distributed inputs, the control effort is spread over some measurable subset of the domain. For boundary/point inputs, the input precisely determines the state at a collection of points of zero measure. An example of a distributed input is RF heating of a plasma in a tokamak~\cite{bribiesca2011strict}.  Examples of point actuation include a thermostat in HVAC regulation or the speaker in noise-cancelling headphones. In a similar manner, output may also be classified using either distributed or boundary/point measurements. A more subtle distinction is the classification as hyperbolic, parabolic or elliptic - a distinction determined by the number and type of partial derivatives. Additionally, we distinguish between isotropic and anisotropic systems. In isotropic systems, independent variables (spatial or temporal) do not appear in the coefficients, whereas the anisotropic form allows such dependence. Examples of anisotropic systems include heat conduction with non-homogeneous/time-varying conductive properties or a wave propagating through a medium of varying density. Finally, we classify the boundary conditions using terms such as Neumann/Dirichlet/Robin/etc. to denote which boundary points are specified or controlled. Classification of boundary conditions has a significant influence on the existence and mathematical properties of the solution~\cite{lions1972non}.

In this paper we focus on the more difficult case of point actuation of a single-state anisotropic parabolic partial-differential equation in a single spatial variable using point observations and non-homogeneous boundary conditions.

There has been significant recent effort to understand and solve the problem of optimal control for PDE systems of this form. For instance,~\cite{van1993h} solved certain distributed input/distributed output optimal control problems using infinite-dimensional Ricatti equations. Additionally,~\cite{lasiecka2000control} and related work considered the problem of point actuation using Ricatti Equations and also discusses potential numerical methods for solving these equations. In~\cite{lasiecka1994control}, an extension of this approach to output feedback through the use of a Luenberger observer is developed. One relatively popular and practical method for controlling parabolic PDE systems has been backstepping~\cite{krstic2008boundary} and its numerous extensions (e.g.~\cite{krstic2008adaptive,smyshlyaev2007adaptive,smyshlyaev2007adaptive2,smyshlyaev2006lyapunov}). This method is attractive due to its straightforward explanation and implementation. However, it does have drawbacks including suboptimality due to the fixed structure of the controller and Lyapunov function.
Additionally, we note some other recent use of Lyapunov functions for analysis and control of infinite dimensional systems including: a rotating beam~\cite{coron1998stabilization}; quasilinear hyperbolic systems~\cite{coron2008dissipative}; and control of systems governed by conservation laws \cite{coron2007strict}.

Alternatively, Sturm-Liouville theory can also be used to devise stabilizing controllers for the class of PDEs we consider. In particular, the problem of searching for the eigenvalues of the differential operators defining the PDEs under consideration can be cast as a Sturm-Liouville eigenvalue problem. Thus, the eigenvalues of the differential operators can be found and consequently, stability properties can be inferred. Moreover, using the same approach, static output feedback controllers which stabilize the PDEs can also be found. Albeit relatively more complicated, the methodology presented in this work has various advantages over the Sturm-Liouville approach, chiefly among which is that we use Lyapunov functionals to achieve our results. Due to this, the presented work can be generalized to construct robust controllers for not only the systems under consideration, but also for nonlinear and uncertain PDEs. Moreover, the numerical examples provided in the paper show that the presented methodology is more effective in constructing stabilizing controllers.

To summarize, although there are a number of methods for control of PDEs, none of them are an ideal solution in the sense that if a controller exists, we have a practical and numerically efficient way to find it. Some previous work in this direction includes the use of Sum-of-Squares for stability analysis of nonlinear PDEs in~\cite{papachristodoulou2005constructing} and was applied to fluid-flow in~\cite{tanaka2009sum}. Additionally, the use of LMIs for stability analysis of semilinear parabolic and hyperbolic systems can be found in \cite{fridman2009lmi}.
The results presented in this paper are a further step towards that ideal solution in the sense that the conditions are convex (meaning they are tractable) and asymptotically accurate (meaning that for any desired accuracy, we can find a convex set of conditions).

Specifically, in this paper, we consider a linear 1-D parabolic partial-differential equation with spatially- and temporally-varying coefficients. We focus on point actuation of Neumann-type boundary conditions, although the use of Dirichlet, Robin, or distributed inputs is also discussed. Our approach to controller synthesis is to use the semigroup framework to formulate the controller synthesis problem as a set of linear operator inequalities. These operator inequalities represent the conditions for existence of a decreasing quadratic Lyapunov function. For point observation, we use the Luenberger observer framework to construct additional inequalities which define the observer. Once we have defined our operator inequalities, we parameterize the set of solutions using operators with polynomial multipliers and kernels. This parametrization is convex and can be tested using recently developed methods for the optimization of positive polynomials such as Sum-of-Squares~ \cite{prajna2001introducing}. Some illustrative examples are also included. The results in this paper fall short of the ideal solution in that they rely on the Luenberger observer for state estimation - meaning the closed loop system may be suboptimal. In addition, the results in this paper cannot be directly applied to vector-valued PDE systems.

\section{Notation}\label{Prelim}
The set $\R^{m \times n}$ contains real matrices of dimensions $m$-by-$n$. The set $\S^n$ contains real symmetric matrices of dimension $n$-by-$n$. $C^1[X]$ is the space of continuously differentiable functions defined on $X$. The shorthand $u_x$ denotes the partial derivative of $u$ with respect to independent variable $x$. $(L_2[X])^n$ denotes the Hilbert space of Lebesgue measurable maps from $X$ to $\R^n$.
$I_n$ is the identity matrix of dimension $n \times n$ and we denote $I=I_n$ when $n$ is clear from context.
We define $Z_d(x)$ to be the vector of monomials in variables $x$ of degree $d$ or less. We define $Z_{n,d}(x)=I_n \otimes Z_d(x)$ - the polynomial matrix whose rows form a basis for vector-valued polynomials of degree $d$ or less.

Unless otherwise indicated, $\langle \cdot,\cdot \rangle$ denotes the inner product on $L_2$ and $\|\cdot\|=\|\cdot\|_{L_2}$ denotes the norm induced by the inner product.
The Sobolev subspace of differentiable functions 
\begin{align*}H^n(0,1):=\{& y \in L_2 \,:y,\cdots,\frac{d^{n-1}y}{dt^{n-1}} \text{ are absolutely continuous} \\
& \text{ with } \frac{d^{n}y}{dt^{n}} \in \lt \}
\end{align*} is equipped with inner product $\ip{x}{y}_{H^n} = \sum_{m=0}^n\ip{\frac{d^m x}{dt^m}}{\frac{d^m y}{dt^m}}$. For Hilbert spaces $X$ and $Y$, the set $\mcl{L}(X,Y)$ includes bounded linear operators from $X$ to $Y$ endowed with the induced norm $\|\cdot\|_\mathcal{L}$.

\section{Background}\label{proset}
In this paper, we focus on the following class of parabolic PDE.
\begin{align}
 &\label{eqn:PDE_form}w_t(x,t)=a(x)w_{xx}(x,t)+b(x)w_x(x,t)+c(x)w(x,t),\\
 & x \in [0,1], \quad t\ge 0 \nonumber
\end{align} with mixed boundary conditions of the form
\begin{equation}
\label{eqn:PDE_form_BC}
 w(0,t)=0, \qquad w_x(1,t)=u(t).
\end{equation}
 For this paper, we assume $w$ is scalar-valued ($w(x,t) \in \mathbb{R}$). Additionally, we assume that $a$, $b$ and $c$ are known polynomial functions with $a(x) \geq \alpha > 0$, for $x \in [0,1]$. Note that the results of this paper can be readily modified to cover Dirichlet, Neumann or Robin-type boundary conditions or systems with time-varying uncertainty in the coefficients.
In addition, note that conditions for well-posedness of this model under feedback have been established in, e.g.~\cite{triggiani1980well,lasiecka1983feedback,lasiecka1980unified,lions1972non}.

In this paper, we will consider state-feedback of the form $u(t)=(\mathcal{F}w)(t)$ where $\mathcal{F}: H^1(0,1) \rightarrow \mathbb{R}$ is a bounded linear operator. It has been shown~\cite{fridman2009lmi} that such feedback is well-posed with a unique local strong solution for any initial condition $w(x,0)=w_0(x) \in \mathcal{D}$, where we define the space
\begin{equation}
 \label{state_space}
 \mathcal{D}=\{z \in H^2(0,1): \quad z(0)=0, \quad z_x(1)=\mathcal{F}z\}.
\end{equation} For the purposes of stability analysis, we also define
\begin{equation}
 \label{state_space_stab}
 \mathcal{D}_0=\{z \in H^2(0,1): \quad z(0)=0, \quad z_x(1)=0\}.
\end{equation}

\subsection{Sum-of-Squares Polynomials (SOSPs)}\label{SOS_sec}

Sum-of-Squares (SOS) is an approach to  the optimization of positive polynomial variables. A typical formalism for the polynomial optimization problem is given by
\begin{align*}
&\max_x \; \; c^T x,\quad \text{subject to }\quad \sum_{i=1}^m x_i f_i(y)  + f_0(y) \ge 0,
\end{align*}
for all $y\in \R^n$, where the $f_i$ are real polynomial functions. The key difficulty is that the feasibility problem of determining whether a polynomial is globally positive ($f(y)\ge0$ for all $y\in \R^n$) is NP-hard~\cite{blum1998complexity}. To overcome this difficulty, there are a number of sufficient conditions for polynomial positivity. A particularly important such condition is that the polynomial, $p$, be a Sum-of-Squares,
\[
p\,(x)=\sum_{i=1}^k g_i(x)^2,
\]
where the $g_i$ are polynomials and which is denoted $p \in\Sigma_s$. The importance of the SOS condition lies in the fact that it can be readily enforced using semidefinite programming. This is due to the easily proven fact that for a polynomial $p$ of degree $2d$, $p \in \Sigma_s$ if and only if $p=Z(x)^T Q Z(x)$ for some $Q\ge 0$, where $Z(x)$ is the vector of monomials of degree $d$ or less. In this way, optimization of positive polynomials can be converted to semidefinite programming.
The semidefinite-programming approach to polynomial positivity was described in the thesis work of~\cite{parrilo2000structured} and also in~\cite{powers1998algorithm}. See also~\cite{chesi1999convexification} and~\cite{lasserre2001global} for contemporaneous work. MATLAB toolboxes for manipulation of SOS variables have been developed and can be found in~\cite{prajna2001introducing} and~\cite{henrion2003gloptipoly}.

SOS can also be used to optimize polynomials which are positive on a subset of $\R^n$ via Positivstellensatz (PS) results~\cite{stengle1974nullstellensatz,
schmudgen1991thek,putinar1993positive,jacobi2001representation}. To see this, consider a semialgebraic set
\begin{equation}
X:=\{x \in \R^n\,:\, g_i(x)\ge 0,\;\; i=1,\cdots,k\}\label{eqn:semialgebraic_set}
\end{equation}
for polynomials $g_i$. A simplified form of PS result can be derived from~\cite{putinar1993positive} and summarized as follows.
\begin{theorem}
For given polynomials $g_i$, suppose that $X$ is defined as per Equation~\eqref{eqn:semialgebraic_set}. Further suppose that $\{x\,:\, g_i(x) \ge 0\}$ is compact for some $i$. If the polynomial $f$ satisfies $f(x)>0$ for $x \in X$, then there exist Sum-of-Squares polynomials $s_i \in \Sigma_s$ such that
\begin{equation*}
f(x) = s_0(s) + \sum_{i=1}^m s_i(s)g_i(s)
\end{equation*}
\end{theorem}
As an illustration of this result, suppose we can find Sums-of-Squares polynomials $s_0$ and $s_1$, such that $p\,(x)=s_0\,(x)+ (1-x^2) \; s_1\,(x)$. Then $p(x)\ge 0$ for $x^2\le 1$. The PS tells us that if $p$ is strictly positive ($p(x)\ge \epsilon>0$ for $x^2 \le 1$), then such polynomials $s_0$ and $s_1$ will always exist. A summary of PS results can be found in~\cite{scheiderer2009positivity}.

\section{A Framework for Analysis and Synthesis of PDEs}

The goal of this paper is to create a practical framework for controller synthesis akin to the LMI framework for ordinary differential equations. To motivate this approach, we recall some notation from the well-developed field of Semigroup theory discussed in the introduction. Within the semigroup framework are certain classes of systems which admit a continuously parameterized operator $S(t)$ which represents the solution map so that any solution $w(t)$ satisfies $S(s)w(t) = S(t+s)w(0)$. Associated with such systems is a possibly unbounded operator $\mathcal{A}:X\rightarrow Y$ known as the infinitesimal generator which satisfies $\dot w(t) = \mathcal{A}w(t)$ for any $w(t)=S(t)w(0)$ where $X$ and $Y$ are Hilbert spaces which depend on the system.

Although we do not explicitly use semigroup theory in this paper, it provides a convenient shorthand for presenting and interpreting our results. Specifically, for PDEs in the form of Equation~\eqref{eqn:PDE_form}, we define the first-order differential form
\begin{equation}
\dot w(t) = \mathcal{A}w(t)+\mathcal{B}u(t) \label{eqn:firstorderform}
\end{equation}
where the operator $\mcl{A}:\mcl{D}_0 \subset \lt \rightarrow \lt$ is defined as
\begin{equation}
(\mathcal{A}w)(x):= a(x)\frac{d^2}{dx^2}w(x)+b(x)\frac{d}{dx}w(x)+c(x)w(x)\label{eqn:Aoperator},
\end{equation} and the space $\mcl{D}_0$ has been defined in Equation~\eqref{state_space_stab}.
Moreover, analogous to the examples in \cite{van1993h} and \cite{byrnes1999example}, it can be established that
\begin{align*}
&(\mathcal{B}u(t))(x)=\delta_1(x)u(t) \quad \text{and} \\
& y(t)=\mcl{C}w(t)=\ip{\delta_1(\cdot)}{w(t)}=w(1,t),
\end{align*} where $\delta_1$ is the Dirac delta functional centered at $x=1$. It can be established that the operator $\mcl{A}$, with domain $\mcl{D}_0$, generates a strongly continuous semigroup $S(t)$ on $\lt$ \cite{curtain1995introduction}. Let $\mcl{D}_1=\mcl{D}_0$ with the norm $\norm{x}_1=\norm{(\alpha I -\mcl{A})x}$, $\alpha \in \rho(\mcl{A})$, where $\rho(\mcl{A})$ is the resolvent set of $\mcl{A}$. Additionally, let $\mcl{D}_{-1}$ be the completion of $\lt$ with respect to the norm $\norm{x}_{-1}=\norm{(\alpha I -\mcl{A})^{-1}x}$. Then, it has been shown in \cite{harkort2011discrete}, using the results presented in \cite{nagel2000one} and \cite{weiss1989representation}, that $\mcl{B} \in \mcl{L}(\R,\mcl{D}_{-1})$ and $\mcl{C} \in \mcl{D}_1^\star$, the dual space of $\mcl{D}_1$. Additionally, it has been shown that Equation~\eqref{eqn:firstorderform} has a continuous state strong solution $w(\cdot) \in C([0,\infty];\lt)$ for any $u \in L_2([0,T];\lt)$, for all $0 < T <\infty$.  

One of the advantages of the operator framework associated with the Semigroup approach is a simplified treatment of Lyapunov functions. Specifically, it is known~\cite{curtain1995introduction} that the strongly continuous semigroup $S(t)$ generated by $\dot{w}=\mathcal{A}w$ is exponentially stable if and only if there exists a positive operator $\mathcal{P}:X\rightarrow X$ such that
\begin{equation}
\ip{\mathcal{A}w}{\mathcal{P}w}_X+\ip{\mathcal{P}w}{\mathcal{A}w}_X\le -\norm{w}^2.\label{eqn:lyapunov_LOI}
\end{equation}

We refer to the feasibility of Condition~\eqref{eqn:lyapunov_LOI} as a \emph{Linear Operator Inequality} (LOI). This condition in particular is equivalent to the existence of a decreasing Lyapunov function of the form $V(w) = \ip{w}{\mathcal{P}w}_X$. Of course, there have been many Lyapunov stability tests proposed in the literature for analysis of infinite-dimensional systems. The goal of this paper, however, is to extend these results to controller and observer synthesis.

Roughly speaking, the approach we take in this paper is to formulate linear operator inequalities similar to Condition~\eqref{eqn:lyapunov_LOI} and interpret these inequalities using Lyapunov functions of the form $V(w) = \ip{w}{\mathcal{P}w}$ where the operator $\mathcal{P}$ is parameterized using polynomials. Positivity is enforced using Sum-of-Squares and the results in~\cite{peetlmi}. The sections in this paper are defined by the particular form of LOI problem which we hope to solve. Specifically, we have the following problems.
\begin{enumerate}
\item Stability
\[
\ip{\mathcal{A}w}{\mathcal{P}w}+\ip{\mathcal{P}w}{\mathcal{A}w} \le -\epsilon \norm{w}^2, 
\]
\item Controller Synthesis
\begin{equation}
\label{eqn:LOI_control}
\ip{(\mathcal{A}\mathcal{P}+\mathcal{B}\mathcal{Z}) w}{w}+\ip{w}{(\mathcal{A}\mathcal{P}+\mathcal{B}\mathcal{Z}) w} \le -\epsilon \norm{w}^2, 
\end{equation}
\item Observer Synthesis
\begin{equation}
\label{eqn:LOI_observer}
\ip{(\mathcal{P}\mathcal{A}+\mathcal{V}\mathcal{C})w}{\mathcal{P}w}+\ip{w}{(\mathcal{P}\mathcal{A}+\mathcal{V}\mathcal{C}) w} \le -\epsilon \norm{w}^2, 
\end{equation}
\end{enumerate} for $w \in \mcl{D}_0$.
In the inequalities above, $\mcl{A}$, $\mcl{B}$ and $\mathcal{C}$ are as defined previously.
Furthermore, we parameterize the operators $\mathcal{P}$, $\mathcal{Z}$ and $\mathcal{V}$ as follows.
\begin{align}
(\mathcal{P}w)(x)=M(x)w(x) +  &\int_0^x K_1(x,y) w(y) d y \nonumber \\
 +  &\label{eqn:operator}\int_y^1 K_2(x,y) w(y) d y, 
\end{align}
where $M(x):[0,1] \rightarrow \mathbb{S}^n$ and $K_1(x,y),K_2(x,y): [0,1] \times [0,1] \rightarrow \mathbb{R}^{n \times n}$ are polynomial matrices and $w \in \lt^n$.

The operator $\mathcal{Z}:H^1(0,1) \rightarrow \R$ is parameterized using  $R_1\in \R$ and  polynomial $R_2$ as
\begin{equation}
\mathcal{Z}w:= R_1 w(1) + \int_0^1 R_2(y)w(y)dy.
\end{equation}

The operator $\mathcal{V}:\R \rightarrow \lt$ is parameterized using polynomial $G_0$ as
\begin{equation}
\left(\mathcal{V}r\right)(y):=G_0(y)r.
\end{equation}
\section{Positive operators and semi-separable polynomial kernels}\label{subsec:positivity}

In this paper, our results are expressed as optimization over a set of positive operators. To solve these optimization problems, we use positive matrices to parameterize a subset of positive operators on $(\lt)^n$ as described in~\cite{peetlmi}. We consider operators of the form
\begin{align}
(\mathcal{P}w)(x)=M(x)w(x) +  &\int_0^x K_1(x,y) w(y) d y \nonumber \\
 +  &\label{eqn:Poperator}\int_x^1 K_2(x,y) w(y) d y, 
\end{align}
where $M(x):[0,1] \rightarrow \mathbb{S}^n$ and $K_1(x,y),K_2(x,y): [0,1] \times [0,1] \rightarrow \mathbb{R}^{n \times n}$ are polynomial matrices and $w \in \lt^n$.  In~\cite{peet2008using}, we gave necessary and sufficient conditions for positivity of multiplier and integral operators of similar form using pointwise constraints on the functions $M$, $K_1$ and $K_2$. Recently, in~\cite{peetlmi}, these conditions was sharpened - See Theorem~\ref{thm:jointpos}.

\begin{theorem}\label{thm:jointpos}
Given $d_1, d_2, n \in \mathbb{N}$ and $\epsilon \in \mathbb{R}$, $\epsilon > 0$, let $Z_1(x) = Z_{n,d_1}(x)$ and $Z_2(x,y) = Z_{n,d_2}(x,y)$.
Suppose there exists a matrix $U$ such that
\[U=\left[\begin{array}{ccc} U_{11}-\epsilon I & U_{12} & U_{13} \\
\star & U_{22} & U_{23} \\
\star & \star & U_{33}
\end{array} \right] \ge 0,\] where the $U_{ij}$ are a partition of $U$. Let
\begin{align*}
&M(s) = Z_{1}(x)^T Q_{11}Z_{1}(x),\\
&K_1(x,y) = Z_{1}(x)^T U_{12}Z_{2}(x,y) + Z_{2}(y,x)U_{31}Z_1(y)\\
&+\int_0^y Z_{2}(\theta,x)^T U_{33}Z_{2}(\theta,y)d\theta  +\int_y^x Z_{2}(\theta,x)^T U_{32}Z_{2}(\theta,y)d\theta \\
&+\int_x^1 Z_{2}(\theta,x)^T U_{22}Z_{2}(\theta,y)d\theta,
\end{align*}
and
\begin{align*}
&K_2(x,y) = Z_{1}(x)^T U_{13}Z_2(x,y) + Z_{2}(y,x)U_{21}Z_{1}(y)\\
&+\int_0^x Z_{2}(\theta,x)^T U_{33}Z_{2}(\theta,y)d\theta+  \int_x^y Z_{2}(\theta,x)^T U_{23}Z_{2}(\theta,y)d\theta \\
& +\int_y^1 Z_{2}(\theta,x)^T U_{22}Z_{2}(\theta,y)d\theta .
\end{align*}

Then the operator $\mathcal{P}$, defined by Equation~\eqref{eqn:Poperator} is  self-adjoint and satisfies
\[\langle \mathcal{P}w,w \rangle \geq \epsilon \|w\|^2, \text{ for all } w \in \lt^n.\]

\end{theorem}
\begin{proof}
See~\cite{peetlmi} for a proof.
\end{proof}

For convenience, we define the set of multipliers and kernels which satisfy Theorem~\ref{thm:jointpos}.
\begin{align*}
 \Xi_{\{d_1,d_2,\epsilon\}} =\{& M,K_1,K_2 \, : \, M,K_1,K_2 \text{ satisfy the conditions of} \\
 & \text{ Theorem~\ref{thm:jointpos} for $d_1,d_2,\epsilon$.}\}
\end{align*}

\section{Inverses of Positive Operators}

As is the case for the finite-dimensional equivalents of Operator Inequalities~\eqref{eqn:LOI_control} and~\eqref{eqn:LOI_observer}, reconstruction of the controller ($u=\mathcal{F}w$) and observer ($\dot{\hat w} = \mathcal{A}\hat w + \mathcal{O}(\hat y - y)$) from a feasible solution of the LOI requires inversion of the operator $\mathcal{P}$ as $\mathcal{F} = \mathcal{Z}\mathcal{P}^{-1}$ and $\mathcal{O} = \mathcal{P}^{-1}\mathcal{V}$. Thus, if we are to use the parametrization of positive operators described in Section~\ref{subsec:positivity}, then given such a positive operator, we must have a reliable way of finding its inverse. For operators without joint positivity, this procedure has been presented in~\cite{peet2009inverses} and expanded in~\cite{peet2013inverses}. In this subsection, we further expand these results by proposing a numerical method for constructing inverses for the class of operators considered in Subsection~\ref{subsec:positivity}.
 Specifically, for scalar valued polynomials $M(x)$, $K_1(x,\xi)$ and $K_2(x,\xi)$ which satisfy the conditions of Theorem~\ref{thm:jointpos}, we will provide a method to construct $\mathcal{P}^{-1}$.

Naturally, all positive operators in the sense of Theorem~\ref{thm:jointpos} are invertible. Our approach is to use a power series expansion with terms which are readily constructed from the matrices described in Theorem~\ref{thm:jointpos}. A closely related result for operators which consist of the identity plus a Volterra operator can be found in \cite[Sec 1.99]{shilov1974elementary}. Our case is slightly different in that we have a positive multiplier  and the Volterra operator is combined with its transpose. Note that the conditions of this theorem are very conservative. In our experience, the series converges whenever $\mathcal{P}$ is positive.

\begin{theorem}\label{thm:inv_op}
Suppose $\{M,K_1,K_2\} \in \Xi_{d_1,d_2,\epsilon}$ for some $d_1,d_2 \in \mathbb{N}$ and $\epsilon >0$. Additionally assume that
\[ |K_1(x,y)| < \epsilon \quad \text{and}\quad  |K_2(x,y)| < \epsilon\quad \text{for all } (x,y) \in [0,1] \times [0,1].\] Then for the operator $\mathcal{P}$ defined as $\mathcal{P}=\mathcal{T}+\mathcal{S}$, where
\begin{align*}
(\mathcal{T}w)(x)=&M(x)w(x)\text{ and } \\
(\mathcal{S}w)(x)=&\igzx K_1(x,y)w(y)dy+\igxo K_2(x,y)w(y)dy,
\end{align*} the inverse is given by
\[\mathcal{P}^{-1}=\left( \sum_{k=0}^\infty(-\mathcal{T}^{-1}\mathcal{S})^k  \right)\mathcal{T}^{-1},\] where
\[(\mathcal{T}^{-1}w)(x)=M(x)^{-1}w(x).\]
\end{theorem}
\begin{proof}
We begin by noting that since $M,K_1,K_2 \in  \Xi_{d_1,d_2,\epsilon}$, $M(x) \geq \epsilon >0$ for all $x \in [0,1]$. Thus
\[
(\mathcal{T}^{-1}w)(x)=M(x)^{-1}w(x), \text{ for all }w \in \lt.
\]
Consequently,
$\mathcal{P}=\mathcal{T}+\mathcal{S}=\mathcal{T}(I+\mathcal{T}^{-1}\mathcal{S})$ is well defined. The small-gain theorem states that  if $\|\mathcal{T}^{-1}\mathcal{S}\| <1 $ then $(I+\mathcal{T}^{-1}\mathcal{S})^{-1}$ exists, is bounded and is given by the convergent series
\[
(I+\mathcal{T}^{-1}\mathcal{S})^{-1}= \sum_{k=0}^\infty(-\mathcal{T}^{-1}\mathcal{S})^k.
\]
First we examine $\mathcal{T}^{-1}$.
\begin{align}
\|\mathcal{T}^{-1}\|&=\sup_{\|w\|=1}|\ip{\mathcal{T}^{-1}w}{w}|=\sup_{\|w\|=1} \left|\igzo \frac{w(x)^2}{M(x)}dx \right| \nonumber \\
&\label{eqn:inv_op_1}\leq \frac{1}{\epsilon} \sup_{\|w\|=1} \igzo w(x)^2 dx = \frac{1}{\epsilon}.
\end{align}
Now, looking at $\mathcal{S}$,
\begin{align*}
\|\mathcal{S}\|= \sup_{\|w\|=1} \bigg| &\igzo \igzx w(x)K_1(x,y)w(y) dy dx   \\
  + &\igzo \igxo w(x)K_2(x,y)w(y) dy dx  \bigg| \\
\leq  \sup_{\|w\|=1} \bigg( &\igzo \igzx |w(x)||K_1(x,y)||w(y)| dy dx   \\
 + &\igzo \igxo |w(x)||K_2(x,y)||w(y)| dy dx  \bigg) .
\end{align*}
By hypothesis we have that $|K_1(x,y)| < \epsilon$ and $|K_2(x,y)| < \epsilon$ and from the triangle, submultiplicative and Holder inequalities we have
\begin{align}
\|\mathcal{S}\| <\,  \epsilon \sup_{\|w\|=1}  \bigg( &\igzo \igzx |w(x)||w(y)| dy dx  \nonumber \\
  + &\igzo \igxo |w(x)||w(y)| dy dx \bigg) \nonumber \\
= \, \epsilon \sup_{\|w\|=1}  \bigg( &\igzo \igzo |w(x)||w(y)| dy dx  \bigg) \nonumber \\
= \, \epsilon \sup_{\|w\|=1}  \bigg( &\igzo |w(x)| dx  \bigg)^2 \nonumber \\
\leq \, \epsilon \sup_{\|w\|=1} \quad &\label{eqn:inv_op_2} \igzo (w(x))^2 dx = \epsilon.
\end{align} Thus from \eqref{eqn:inv_op_1} and \eqref{eqn:inv_op_2},
\[\|\mathcal{T}^{-1}\mathcal{S}\| \leq \|\mathcal{T}^{-1}\|\|\mathcal{S}\|<1.\] Hence $(I+\mathcal{T}^{-1}\mathcal{S})^{-1}= \sum_{k=0}^\infty(-\mathcal{T}^{-1}\mathcal{S})^k$, which implies
\[\mathcal{P}^{-1}=(\mathcal{T}+\mathcal{S})^{-1}=(I+\mathcal{T}^{-1}\mathcal{S})^{-1} \mathcal{T}^{-1}= \left( \sum_{k=0}^\infty(-\mathcal{T}^{-1}\mathcal{S})^k \right) \mathcal{T}^{-1}.\]
\end{proof}

For convenience, we define the set of multipliers and kernels which satisfy the conditions of both Theorem~\ref{thm:jointpos} and Theorem~\ref{thm:inv_op}.
\begin{align*}
\Omega_{d_1,d_2,\epsilon}=&\{M,K_1,K_2:M,K_1,K_2 \in \Xi_{d_1,d_2,\epsilon} \text{ and satisfy the} \\
& \text{ conditions of Theorem~\ref{thm:inv_op} for }d_1,d_2,\epsilon \}.
\end{align*}

To construct the inverse, then, we use the MuPAD symbolic engine of MATLAB to evaluate the series $\left(\sum^K_{k=0}\left(-\mathcal{T}^{-1}\mathcal{S}\right)^k\right)\mathcal{T}^{-1}$ for some finite $K$ where $K$ is chosen sufficiently large so that the series adequately approximates the inverse. In practice, we have found that only a few terms are required for convergence. To illustrate, in Figures~\ref{fig:opinverse1} and \ref{fig:opinverse2} we find some $(M,K_1,M_2) \in \Omega_{2,2,2}$ and find $\mathcal{P}_K^{-1}=\left(\sum^K_{k=0}\left(-\mathcal{T}^{-1}\mathcal{S}\right)^k\right)\mathcal{T}^{-1}$
for several values of $K$. Then we plot $\|w-\mathcal{P}\mathcal{P}_K^{-1}w\|_{L_2}$ and $\|w-\mathcal{P}_K^{-1}\mathcal{P}w\|_{L_2}$ as a function of $K$ for the arbitrarily chose function $w(x)=\sin(5 \pi x)/(x+1)$. In this case, $K=10$ yields norm error of order $\approx 10^{-12}$.

\begin{figure}[ht]
\centering
\subfigure[$\|w-\mathcal{P}\mathcal{P}_K^{-1}w\|_{L_2}$]{%
\includegraphics[scale=0.14]{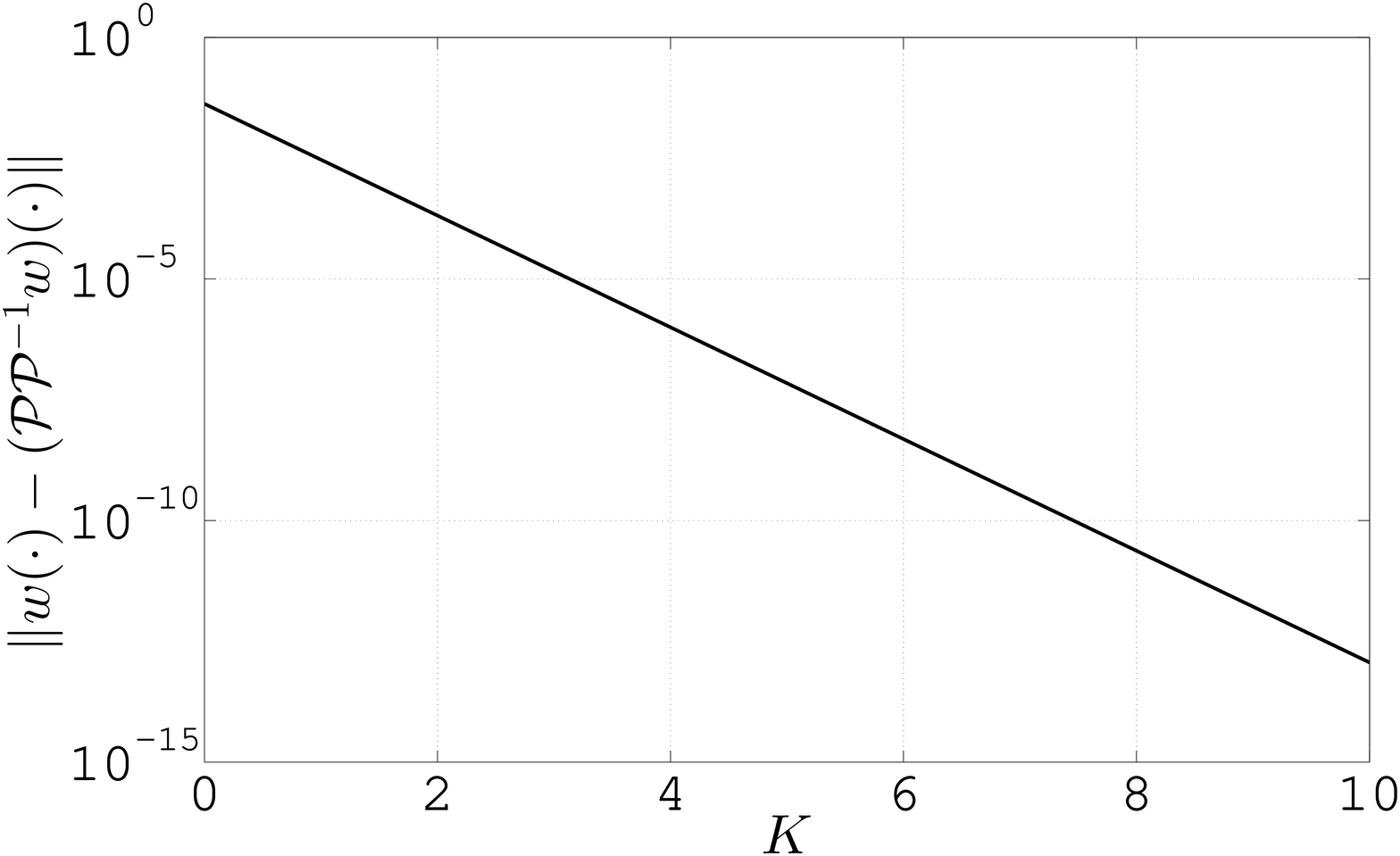}
\label{fig:opinverse1}}
\quad
\subfigure[$\|w-\mathcal{P}_K^{-1}\mathcal{P}w\|_{L_2}$ ]{%
\includegraphics[scale=0.14]{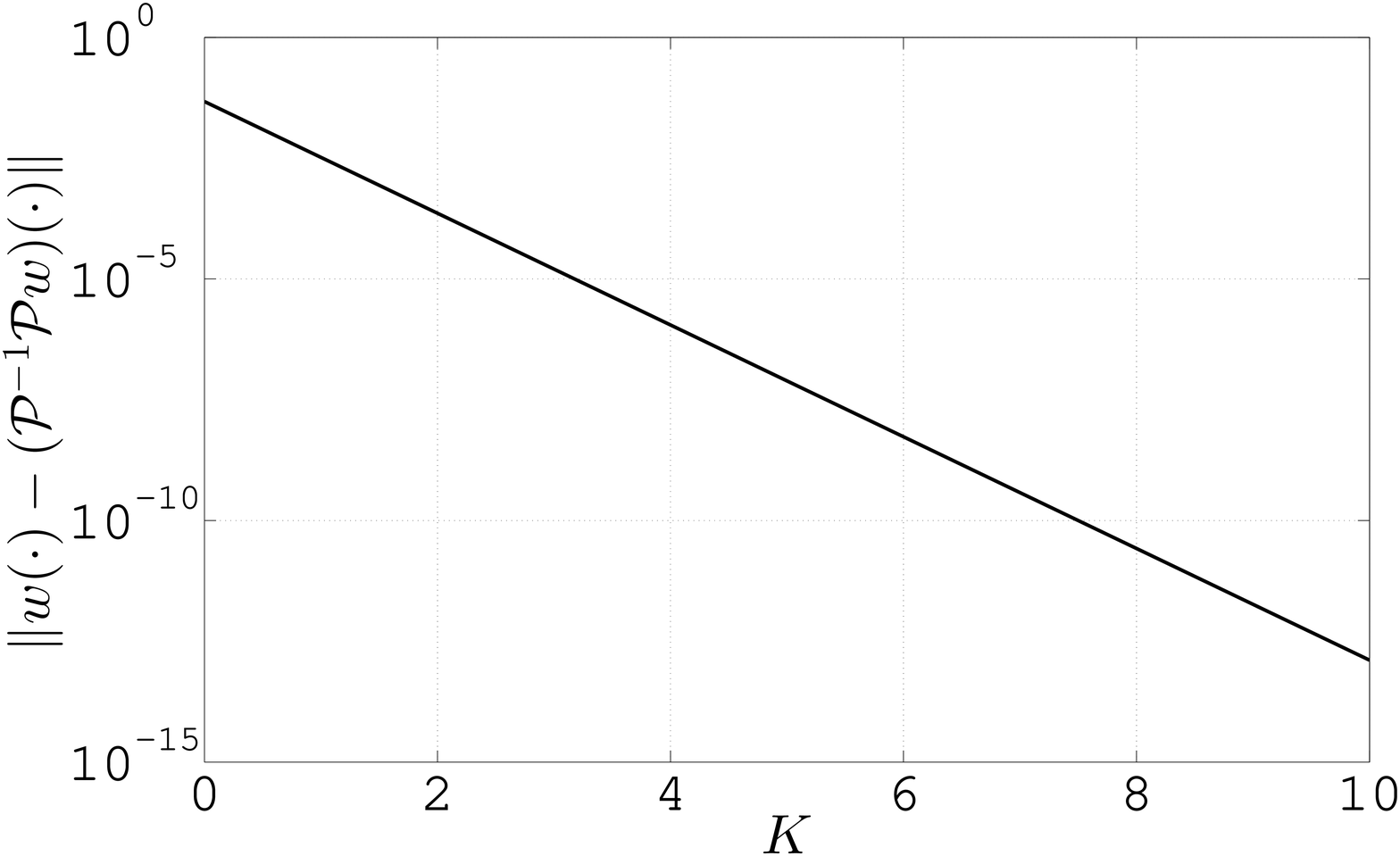}
\label{fig:opinverse2}}
\caption{$\|w-\mathcal{P}\mathcal{P}_K^{-1}w\|_{L_2}$ and $\|w-\mathcal{P}_K^{-1}\mathcal{P}w\|_{L_2}$  as a function of $K$.}
\end{figure}

\section{Stability Analysis}\label{sec:stability}

In this section, we address the simpler problem of stability of PDE systems of the Form~\eqref{eqn:PDE_form}. Roughly speaking, we are looking for a positive operator in the form of Equation~\eqref{eqn:Poperator} which satisfies the inequality
\[
\ip{\mathcal{A}x}{\mathcal{P}x}+\ip{x}{\mathcal{P}\mathcal{A}x} \le -\epsilon \norm{x}^2
\]
for all $x \in \mathcal{D}_0$ where the operator $\mathcal{A}$ is defined in Equation~\eqref{eqn:Aoperator}. The main result relies primarily on the following upper-bound - the proof of which is included in the appendix.
\begin{align}
  \ip{\mathcal{A}x}{\mathcal{P}x}+\ip{x}{\mathcal{P}\mathcal{A}x} \le &  \ip{\bmat{x(1)\\x}}{ \mathcal{Q} \bmat{x(1)\\x}}_{\R \times L_2} \nonumber \\
  &\label{eqn:stab_LOI}+\igzo x_s(0)Q_3(s)x(s)ds,
 \end{align}
where we define the operator $\mathcal{Q}$ as
\begin{align*}
(\mathcal{Q}y)(s) := Q_0(s)\bmat{y(1) \\ y(s)} +   &\int_0^s  \bmat{0 & 0\\ 0 & Q_1(s,t)}\bmat{y(1) \\ y(t)} dt \\
 + &\int_s^1  \bmat{0 & 0\\ 0 & Q_2(s,t)} \bmat{y(1) \\ y(t)} dt,
\end{align*}
where $\{Q_0,Q_1,Q_2,Q_3\}=\mathcal{M}_\epsilon(M,K_1,K_2)$ and where the linear operator $\mathcal{M}_\epsilon$ is defined as follows.
\begin{definition}\label{def:stab_map}
We say $\{Q_0,Q_1,Q_2,Q_3\}=\mathcal{M}_\epsilon(M,K_1,K_2)$ if the following hold
\begin{align}
 &Q_0(s)_{1,1}=\left[\left(b(1)-a_s(1) \right)M(1)-a(1)M_s(1) \right],\\
 &Q_0(s)_{1,2}=Q_0(s)_{2,1} \nonumber \\
 &=\left[\left(b(1)-a_s(1) \right)K_1(1,s)-a(1)K_{1,s}(1,s) \right], \\
&Q_0(s)_{2,2}=\pfs \left[ \pfs \left[ a(s)M(s)\right]-b(s)M(s) \right]+2M(s)c(s) \nonumber \\
&   +  \left[\pfs \left[2a(s)\left(K_1(s,t)-K_2(s,t) \right) \right] \right]_{t=s}-\frac{\pi^2}{2}\alpha \epsilon ,\\
&Q_1(s,t) \nonumber \\
&=\left(\pfs \left[\pfs \left[a(s)K_1(s,t)\right]-b(s)K_1(s,t) \right]+c(s)K_1(s,t) \right) \nonumber \\
&+\left(\pft \left[\pft \left[a(t)K_1(s,t)\right]-b(t)K_1(s,t) \right]+c(t)K_1(s,t) \right), \\ 
&Q_2(s,t)\nonumber \\
&=\left(\pfs \left[\pfs \left[a(s)K_2(s,t)\right]-b(s)K_2(s,t) \right]+c(s)K_2(s,t) \right) \nonumber \\
&+\left(\pft \left[\pft \left[a(t)K_2(s,t)\right]-b(t)K_2(s,t) \right]+c(t)K_2(s,t) \right)\text{ and }\\
&Q_3(s)=-2a(0)K_2(0,s),
\end{align} where $ K_{1,s}(1,s)=\left[K_{1,s}(s,t)|_{s=1} \right]_{t=s}$.
\end{definition}

\begin{theorem}\label{thm:stability}
Suppose that there exist $\{M,K_1,K_2\} \in \Xi_{d_1,d_2,\epsilon}$ and $\epsilon,\delta>0$ such that
\begin{align*}
 &\left\{-Q_{0_{2,2}}- 2\delta M,-Q_1- 2\delta K_1,-Q_2-  2\delta K_2\right\} \in \Xi_{d_1,d_2,0}, \\
 &Q_{0_{1,1}}=0, \quad Q_{0_{1,2}}=0 \quad \text{and} \quad K_2(0,x)=0,
\end{align*}
where $\{Q_0,Q_1,Q_2,Q_3\}=\mathcal{M}_\epsilon(M,K_1,K_2)$. Then, for any initial condition $w(0) \in \mathcal{D}_0$, the solution $w(x,t)$ of Equations~\eqref{eqn:PDE_form}-\eqref{eqn:PDE_form_BC} with $u(t)=0$ satisfies
\[
\|w(\cdot,t)\|_{L_2} \leq e^{-\delta t} \sqrt{\frac{\langle w_0,\mathcal{P}w_0 \rangle}{\epsilon}}, \quad t > 0,
\]
where
\[
(\mathcal{P}z)(x)=M(x)z(x) + \igzx K_1(x,\xi) z(\xi) d \xi + \igxo K_2(x,\xi) z(\xi) d \xi.
\]
\end{theorem}

\begin{proof}
 Consider the following Lyapunov function $ V(w)= \langle w,\mathcal{P}w \rangle_{L_2}$.
 Taking the derivative along trajectories of the system, we have
 \begin{align*}
 \frac{d}{dt}V(w(t))&= \langle w_t(t),(\mathcal{P}w(t)) \rangle+\langle w(t),(\mathcal{P}w_t(t)) \rangle\\
&=\ip{\mathcal{A}w(t)}{\mathcal{P}w(t)}+\ip{w(t)}{\mathcal{P}\mathcal{A}w(t)}.
\end{align*} Since the initial condition $w(0) \in \mathcal{D}_0$, $w(t) \in \mathcal{D}_0$ exists for all $t\ge 0$. For $\mathcal{P}$ as defined in~\eqref{eqn:Poperator} and $\mathcal{M}_\epsilon$ as defined in Definition~\ref{def:stab_map}, it is shown in the Appendix that if $\{Q_0,Q_1,Q_2,Q_3\}=\mathcal{M}_\epsilon(M,K_1,K_2)$, then
\begin{align*}
\frac{d}{dt}V(w(t))=&\ip{\mathcal{A}w(t)}{\mathcal{P}w(t)}+\ip{w(t)}{\mathcal{P}\mathcal{A}w(t)} \\
\leq&  \ip{\bmat{w(1,t)\\w(\cdot,t)}}{ \mathcal{Q} \bmat{w(1,t)\\w(\cdot,t)}}_{\R \times \lt} \\
&+\igzo w_x(0,t)Q_3(x)w(x,t)dx.
\end{align*} Now, by definition, $Q_3(x)=-2a(0)K_2(0,x)$ and since by assumption $K_2(0,x)=0$, we have $Q_3=0$. Moreover, since $Q_{0_{1,1}}=0$ and $Q_{0_{1,2}}=Q_{0_{2,1}}=0$, we have
\begin{alignat*}{2}
&\frac{d}{dt}V(w(t)) \\
&\leq  \igzo w(x,t) \bigg( Q_0(x)_{2,2}(x)w(x,t) &&+\igzx Q_1(x,s)w(s,t)ds \\
 &  &&+\igxo Q_2(x,s)w(s,t)ds \bigg)dx.  
\end{alignat*}
Since
\[
 \left\{-Q_{0_{2,2}}- 2\delta M,-Q_1- 2\delta K_1,-Q_2-2\delta K_2\right\} \in \Xi_{d_1,d_2,0},
\]
we have that
\begin{align*}
&\igzo w(x,t) \left( Q_0(x)_{2,2}(x)w(x,t)+\igzx Q_1(x,s)w(s,t)ds \right. \\
& \left. +\igxo Q_2(x,s)w(s,t)ds \right)dx \leq -2\delta \ip{w(\cdot,t)}{\mathcal{P}w(\cdot,t)}.
\end{align*}
Hence we conclude that
\[\frac{d}{dt}  V(w(t)) \leq -2\delta V(w(t)), \quad t>0.\]
Integrating in time yields $\langle w(\cdot,t), (\mathcal{P}w) (\cdot,t) \rangle \leq e^{-2 \delta t} \langle w_0, \mathcal{P}w_0 \rangle$ and since, $\{M,K_1,K_2\} \in \Xi_{d_1,d_2,\epsilon}$, we have
 \[
 \epsilon \|w(\cdot,t)\|^2 \leq \langle w(\cdot,t), (\mathcal{P}w) (\cdot,t) \rangle \leq e^{-2 \delta t} \langle w_0, \mathcal{P}w_0 \rangle, \quad t>0
 \]
 which implies
 \[
 \|w(\cdot,t)\|_{L_2} \leq e^{-\delta t} \sqrt{\frac{\langle w_0,\mathcal{P}w_0 \rangle}{\epsilon}}, \quad t > 0.
 \]
 \end{proof}

 \subsection{Stability Analysis Numerical Results}\label{stabannum}

\paragraph*{Example 1} To illustrate the accuracy of the stability test, we perform several numerical experiments. For the first test, we check the conditions of Theorem~\ref{thm:stability} on a system whose stability properties are known a priori - $w_t= w_{xx}+\lambda w$. The system is defined by Equations~\eqref{eqn:PDE_form} -~\eqref{eqn:PDE_form_BC} with $u(t)=0$, $a(x)=1$, $b(x)=0$ and $c(x)=\lambda$, where $\lambda > 0$. The analytic solution to this PDE is given by
\[w(x,t)=\sum_{n=1}^\infty e^{\lambda_n t}\langle w_0, \phi_n \rangle \phi_n(x) , \] where $\lambda_n = \lambda - \frac{(2n-1)^2 \pi^2}{4}$, $\phi_n(x)=\sqrt{2} \sin \left( \frac{2n-1}{2} \pi x\right)$ and $w_0$ is the initial condition. Thus, one can see that the boundary-value problem is stable for $\lambda \in [0, \frac{\pi^2}{4})$.
\begin{table}{}
\begin{center}
    \begin{tabular}{l *{7}{c}}\hline \hline
 & $d=3$ & $4$ & $5$ & $6$ & $7$ \\ \hline
$\delta=0.1$ &  $0.55$ & $2.19$ & $2.35$ & $2.36$ & $2.36$ \\
$\delta=0.01$ & $0.59$ & $2.19$ & $2.448$ & $2.451$ & $2.452$ \\
$\delta=0.001$ & $0.59$ & $2.19$ & $2.457$ & $2.46$ & $2.461$
\end{tabular}
\end{center}
\caption{Maximum $\lambda$ as a function of polynomial degree, $d_1=d_2=d$ for $w_t= w_{xx}+\lambda w$ and different exponential decay rates $\delta$.}
\label{table_analysis}
\end{table}
 Table~\ref{table_analysis} presents the accuracy of Theorem~\ref{thm:stability} when applied to the problem of determination of the maximum stable $\lambda$. Note that an increase in the degree of polynomials $d=d_1=d_2$ increases the accuracy of the test in terms of the maximum detectable stable value of $\lambda$. For degree $7$, we can construct a Lyapunov function which proves stability for $\lambda=2.461$, with $\delta=0.001$, which is $99.74 \%$ of the stability margin $\frac{\pi^2}{4}=2.4674$.
\begin{figure}[ht]
\centering
\includegraphics[scale=0.15]{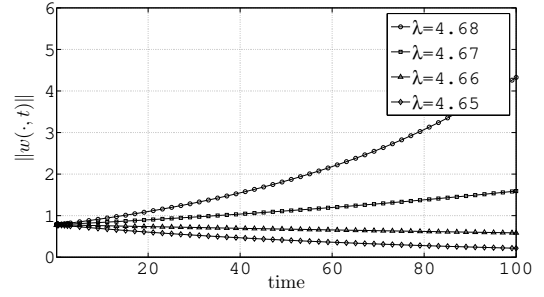}
\caption{State norm evolution for different $\lambda$ for Example 2.}
\label{fig:num_stability}
\end{figure}

\paragraph*{Example 2} For the second numerical test, we consider a completely arbitrary system defined by Equations~\eqref{eqn:PDE_form} -~\eqref{eqn:PDE_form_BC} with $u(t)=0$, $a(x)=x^3-x^2+2$, $b(x)=3x^2-2x$ and $c(x)= -0.5x^3 + 1.3x^2 - 1.5x + 0.7+\lambda$. Again, we seek to determine the maximum value of $\lambda$ for which the system is exponentially stable. The maximum stable $\lambda$ predicted by Theorem~\ref{thm:stability} is shown in Table~\ref{table_analysis_1} for $\epsilon =0.001$. For this system, there is no analytic solution and hence if we wish to determine the accuracy of our results, we must use finite difference methods to simulate the system and hence estimate the true maximum stable value of $\lambda$. This work is presented in Figure~\ref{fig:num_stability}, which suggests that the system is unstable for $\lambda >  4.66 $.
\begin{table}{}
\begin{center}
    \begin{tabular}{l *{7}{c}}\hline \hline
 & $d=3$ & $4$ & $5$ & $6$ & $7$ \\ \hline
$\delta=0.1$ &  $4.27$ & $4.51$ & $4.51$ & $4.52$ & $4.52$ \\
$\delta=0.01$ & $4.36$ & $4.60$ & $4.60$ & $4.61$ & $4.61$ \\
$\delta=0.001$ & $4.37$ & $4.61$ & $4.61$ & $4.62$ & $4.62$
\end{tabular}
\end{center}
\caption{Maximum stable $\lambda$ as a function of polynomial degree for Example 2.}
\label{table_analysis_1}
\end{table}
The maximum $\lambda$ for which we can prove the exponential stability for is $\lambda=4.62$, which is $99.14 \%$ of the predicted stability margin of $4.66$.
\begin{figure}[ht]
\centering
\subfigure[Illustration of $V(t) \geq \epsilon \|w(\cdot,t)\|^2$.]{%
\includegraphics[scale=0.22]{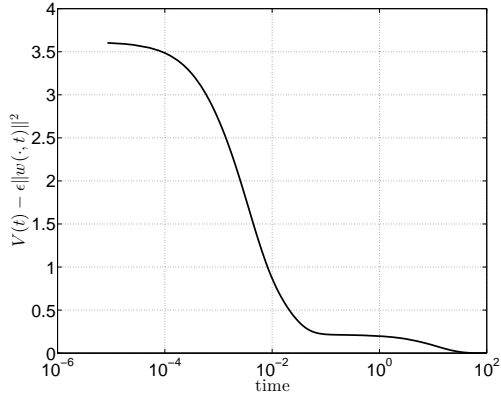}
\label{}}
\quad
\subfigure[Illustration of $\dot{V}(t) \leq -2 \delta V(t) $. ]{%
\includegraphics[scale=0.22]{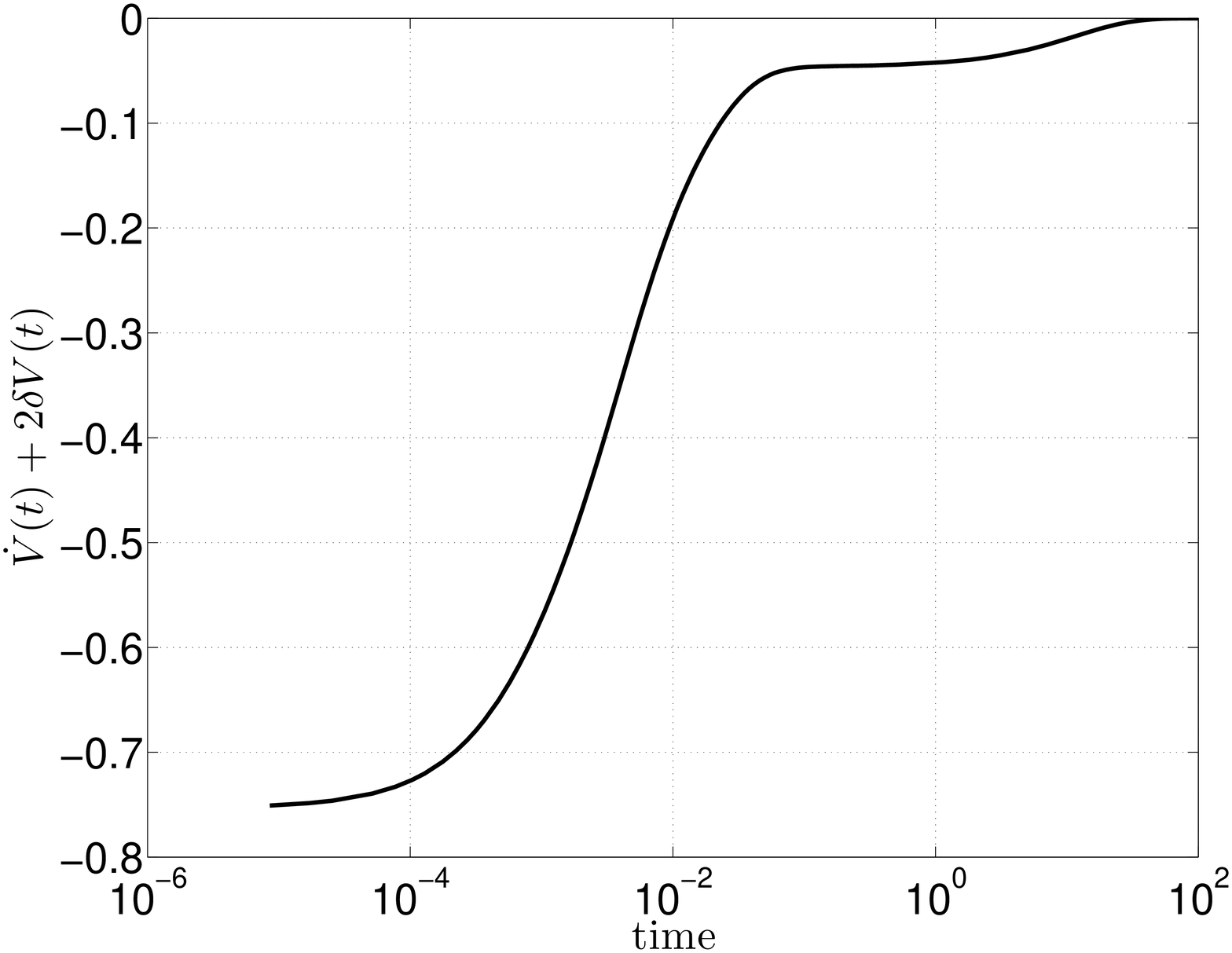}
\label{}}
\caption{Evolution of the Lyapunov functional and its time derivative for $a(x)=x^3-x^2+2$, $b(x)=3x^2-2x$ and $c(x)= -0.5x^3 + 1.3x^2 - 1.5x + 0.7+\lambda$ with $\lambda=4.62$ and $\delta = \epsilon =0.001$.}
\label{fig:V_analysis}
\end{figure}
Finally, although the Lyapunov function generated for this system is too complicated for print, Figure~\ref{fig:V_analysis} illustrates the evolution of this Lyapunov functional time derivative.

\paragraph*{Example 3} In this numerical example we wish to examine if we achieve any performance improvement in the stability analysis by including the integral kernels $K_1$ and $K_2$ in the Lyapunov functional operator $\mathcal{P}$. Thus, we apply Theorem~\ref{thm:stability}, with $K_1=K_2=0$, on the systems considered in Examples 1 and 2. Table~\ref{table_comparison_analysis} presents the results.
\begin{table}{}
\begin{center}
    \begin{tabular}{l *{7}{c}}\hline \hline
 & Example $1$ & Example $2$ \\ \hline
$\lambda$ &  $2.461$ & $4.38$ 
\end{tabular}
\end{center}
\caption{Maximum stable $\lambda$, for $K_1=K_2=0$, for Examples $1$ and $2$ for $\delta=0.001$.}
\label{table_comparison_analysis}
\end{table}
 Comparing Table~\ref{table_comparison_analysis} to Tables~\ref{table_analysis} and \ref{table_analysis_1} shows that for the system considered in Example 1, the integral kernels $K_1$ and $K_2$ do not have an effect. However, the inclusion of $K_1$ and $K_2$ increases the precision in predicting the stability margin for the system considered in Example 2. Thus, this numerical experiment indicates that for systems with distributed coefficients, including $K_1$ and $K_2$ produces sharper results for stability analysis.

\paragraph*{Example 4} For the final numerical test, we wish to examine the effectiveness of the presented method on a system with different boundary conditions. In particular, we consider $w_t=w_{xx}+\lambda w$ with Dirichlet boundary conditions $w(0,t)=w(1,t)=0$. The analytic solution of this PDE can be calculated as
\[w(x,t)=\sum_{n=1}^\infty e^{\lambda_n t} \ip{w_0}{\phi_n}\phi_n(x),\] where $\lambda_n=\lambda-n^2 \pi^2$, $\phi_n(x)=\sqrt{2}\sin (n \pi x)$ and $w_0(x)$ is the initial condition. Thus, the system is stable for $\lambda < \pi^2$. The conditions of Theorem~\ref{thm:stability} can be easily modified to analyze this system. 

\begin{table}{}
\begin{center}
    \begin{tabular}{l *{7}{c}}\hline \hline
 &  $d=4$ & $5$ & $6$ & $7$ & $8$\\ \hline
$\delta=0.1$ &   $1.4$ & $4.9$ & $7.59$ & $9.61$ & $9.7$ \\
$\delta=0.01$ &  $1.5$ & $5.1$ & $7.69$ & $9.63$ & $9.79$ \\
$\delta=0.001$ &  $1.8$ & $5.3$ & $7.99$ & $9.66$ & $9.82$
\end{tabular}
\end{center}
\caption{Maximum $\lambda$ as a function of polynomial degree, $d_1=d_2=d$ for $w_t= w_{xx}+\lambda w$ with Dirichlet boundary conditions and different exponential decay rates $\delta$.}
\label{table_analysis_dirichlet}
\end{table}
 Table~\ref{table_analysis_dirichlet} presents the accuracy of the modified Theorem~\ref{thm:stability} when applied to the problem of determination of the maximum stable $\lambda$ for $w_t=w_{xx}+\lambda w$ with $w(0,t)=w(1,t)=0$. For degree $8$, we can construct a Lyapunov function which proves stability for $\lambda=9.82$, with $\delta=0.001$, which is $99.49 \%$ of the stability margin $\pi^2$.

\section{State-Feedback Controller Synthesis}\label{sec:synthesis}

In this section, we use a dual version of the stability condition in Theorem~\ref{thm:stability} to synthesize full-state feedback controllers. Roughly speaking, the dual stability condition is expressed as the search for a positive operator, $\mathcal{P}$, of the form of Equation~\eqref{eqn:Poperator} which satisfies the inequality
\[
\ip{\mathcal{A}\mathcal{P}x}{x}+\ip{x}{\mathcal{A}\mathcal{P}x} \le -\epsilon \norm{x}^2.
\]
When we include an input of the form $w_x(1,t)=u(t)=\mathcal{F}w(t)$, this becomes
\[
\ip{(\mathcal{A}\mathcal{P}+\mathcal{B}\mathcal{Z}) x}{x}+\ip{x}{(\mathcal{A}\mathcal{P}+\mathcal{B}\mathcal{Z}) x} \le -\epsilon \norm{x}^2
\]
where $\mathcal{F} = \mathcal{Z}\mathcal{P}^{-1}$. Recall the dynamics in Equation~\eqref{eqn:PDE_form}:
\begin{equation}
\label{eqn:PDE_cont_form}
 w_t(x,t)=a(x)w_{xx}(x,t)+b(x)w_x(x,t)+c(x)w(x,t)
\end{equation} with
\begin{equation}
 \label{eqn:PDE_cont_form_BC}
 w(0,t)=0, \qquad w_x(1,t)=u(t)
\end{equation}
with initial condition $w(\cdot,0)=w_0\in \mcl{D}$.
As before our main result uses an upper-bound of the form
\begin{align}
  \ip{\mathcal{A}\mathcal{P}x}{x}+\ip{x}{\mathcal{A}\mathcal{P}x}  \le&  \ip{\bmat{x(1)\\x}}{ \mathcal{T} \bmat{x(1)\\x}}_{\R \times L_2}   \nonumber \\
  &+ x(0)(T_3 x(0)+T_4 x_s(0)),\label{eqn:dual_stab_LOI}
 \end{align}
where the operator $\mathcal{T}$ is defined as
\begin{align*}
(\mathcal{T}y)(s) := T_0(s)\bmat{y(1) \\ y(s)} +   &\int_0^s  \bmat{0 & 0 \\ 0 & T_1(s,t)} \bmat{y(1) \\ y(t)} dt  \\
+ &\int_s^1  \bmat{0 & 0 \\ 0 & T_2(s,t)} \bmat{y(1) \\ y(t)} dt,
\end{align*}
where $\{T_0,T_1,T_2,T_3,T_4\}=\mathcal{N}_\epsilon(M,K_1,K_2)$ and where the linear operator $\mathcal{N}_\epsilon$ is defined as follows.
\begin{definition}
We say $\{T_0,T_1,T_2,T_3,T_4\}=\mathcal{N}_\epsilon(M,K_1,K_2)$ if
\begin{alignat}{2}
T_0(s)_{1,1}&= &&\left[-a(1)M_s(1)+(b(1)-a_s(1))M(1) \right], \\
T_0(s)_{1,2}&=&&T_0(s)_{2,1}=-a(1)K_{1,s}(1,s), \\
T_0(s)_{2,2}&= &&\left[(a_{ss}(s)-b_s(s))M(s)+b(s)M_s(s)\right]+2 M(s)c(s) \nonumber \\
 & && +a(s)\left[M_{ss}(s)+2\pfs \left[K_1(s,t)-K_2(s,t) \right] \right]_{t=s} \nonumber \\
 & && -\frac{\pi^2}{2}\alpha \epsilon ,\\
T_1(s,t)&= &&a(s)K_{1,ss}(s,t)+b(s)K_{1,s}(s,t)+c(s)K_1(s,t) \nonumber \\
 & &&+a(t)K_{1,tt}(s,t)+b(t)K_{1,t}(s,t)+c(t)K_1(s,t), \\
T_2(s,t)&=&& a(s)K_{2,ss}(s,t)+b(s)K_{2,s}(s,t)+c(s)K_2(s,t) \nonumber \\
 & &&+a(t)K_{2,tt}(s,t)+b(t)K_{2,t}(s,t)+c(t)K_2(s,t),\\
T_3&=&&a_x(0)M(0)-a(0)M_x(0)-b(0)M(0)+\frac{\pi^2}{2}\alpha \epsilon \text{ and }\\
T_4&=&&-2 a(0)M(0).
\end{alignat}
\end{definition}
\begin{theorem}[Dual Stability]\label{thm:dualstability}
Suppose there exist $\{M,K_1,K_2\} \in \Omega_{d_1,d_2,\epsilon}$  and $\epsilon,\delta>0$ such that
\begin{align*}
&\left\{-T_{0_{2,2}}-2\delta M,-T_1-2\delta K_1, -T_2-2\delta K_2\right\} \in \Xi_{d_1,d_2,0},\\
&\quad   T_{0_{1,1}}=0,\quad T_{0_{1,2}}=0 \quad \text{and} \quad K_2(0,x)=0,
\end{align*} 
where
$\{T_0,T_1,T_2,T_3,T_4\}=\mathcal{N}_\epsilon(M,K_1,K_2)$.

Then any solution $w$ of~\eqref{eqn:PDE_cont_form} -~\eqref{eqn:PDE_cont_form_BC} with $u(t)=0$ and $w_0 \in \mathcal{D}_0$ satisfies
\[
\|w(\cdot,t) \| \leq \|P\|_\mathcal{L} e^{-\delta t} \sqrt{\frac{\langle w_0,P^{-1}w_0 \rangle}{\epsilon}},
\]
where
 \[
 (\mathcal{P}v)(x)=M(x)v(x) + \igzx K_1(x,\xi) v(\xi) d \xi + \igxo K_2(x,\xi) v(\xi) d \xi.
 \]
 \end{theorem}
The proof of Theorem~\ref{thm:dualstability} will be implied by the proof of Theorem~\ref{thm:synthesis}.

\begin{theorem}[Controller Synthesis]\label{thm:synthesis}
For $\epsilon,\delta>0$, $d_1, d_2 \in \N$, suppose there exist $\{M,K_1,K_2\} \in \Omega_{d_1,d_2,\epsilon}$ such that
\begin{align*}
&\left\{-T_{0_{2,2}}-2 \delta M,-T_1-2 \delta K_1, -W_2-2 \delta K_2\right\} \in \Xi_{d_1,d_2,0}\text{ and } \\
& \quad K_2(0,x)=0,
\end{align*}where
$\{T_0,T_1,T_2,T_3,T_4\}=\mathcal{N}_\epsilon(M,K_1,K_2)$.

Define the operator $\mathcal{F}:=\mathcal{Z}\mathcal{P}^{-1}$ where
\begin{align*}
&(\mathcal{Z}y)=R_1y(1)+\igzo R_2(x)y(x)dx,\\
& R_1=-\frac{T_{0_{1,1}}}{2 a(1)}, \quad R_2=-\frac{T_{0_{1,2}}}{a(1)}.
\end{align*}
Then any solution $w$ of~\eqref{eqn:PDE_cont_form} -~\eqref{eqn:PDE_cont_form_BC} with $u(t)=(\mathcal{F}w)(t)$ and $w_0 \in \mathcal{D}$ satisfies
\[
  \|w(\cdot,t) \| \leq \|\mathcal{P}\|_\mathcal{L} e^{-\delta t} \sqrt{\frac{\langle w_0,\mathcal{P}^{-1}w_0 \rangle}{\epsilon}}, \quad t>0.
\]
 \end{theorem}

\begin{proof}
Consider the following Lyapunov function $V(w)=\ip{w}{\mathcal{P}^{-1}w}$. Taking the time derivative along trajectories of the system, we have
\begin{align*}
\frac{d}{dt}V(w(t))=
 \ip{\mathcal{A}w(t)}{\mathcal{P}^{-1}w(t)}+\ip{\mathcal{P}^{-1}w(t)}{\mathcal{A}w(t)},
\end{align*} where we have used the fact that $\mathcal{P}=\mathcal{P}^\star$ implies $\mathcal{P}^{-1}= {\left(\mathcal{P}^\star\right)}^{-1}$.
Now let $y=\mathcal{P}^{-1}w$. Then $y \in P^{-1}\mathcal{D}$ and
\begin{align*}
\frac{d}{dt}V(w(t))
=& \ip{\mathcal{A}\mathcal{P}y(t)}{y(t)}+\ip{y(t)}{\mathcal{A}\mathcal{P}y(t)}.
\end{align*}

From Corollary~\ref{cor:dual2}, we have
\begin{align}
&\frac{d}{dt}V(w(t))=\ip{\mathcal{A}\mathcal{P}y(t)}{y(t)}+\ip{y(t)}{\mathcal{A}\mathcal{P}y(t)} \nonumber \\
&\leq \ip{\bmat{y(1,t)\\y(\cdot,t)}} {\mathcal{T} \bmat{y(1,t)\\y(\cdot,t)}}_{\R \times \lt} \nonumber \\
&\quad +y(0,t)(T_3y(0,t)+T_4 y_x(0,t)) + 2y(1,t)a(1)M_x(1)y(1,t) \nonumber \\
&\label{eqn:synth_thm_1}\quad + 2y(1,t)a(1) \left(\igzo K_{1,x}(1,x)y(x,t)dx+M(1)y_x(1,t) \right).
\end{align}

Since $w=\mathcal{P}y$, we have
\begin{align*}
w(x,t)=M(x)y(x,t)+&\igzx K_1(x,\xi)y(\xi,t) d\xi \\
+ &\igxo K_2(x,\xi)y(\xi,t) d\xi.
\end{align*}
Thus boundary condition $w(0,t)=0$ and the hypothesis $K_2(0,x)=0$ imply
\begin{align}
&\label{eqn:synth_left_cond} y(0,t)=0.
\end{align}

Similarly, $u(t)=w_x(1,t)$ implies

\[
u(t)=M(1)y_x(1,t)+M_x(1)y(1,t)+\igzo K_{1,x}(1,x)y(x,t)dx.
\]

Combining this with $u(t)=(\mathcal{F}w)(t)=(\mathcal{Z}\mathcal{P}^{-1}w)(t)=(\mathcal{Z}y)(t)$, we obtain

\begin{align}
&(R_1-M_x(1))y(1,t) + \igzo R_2(x)y(x,t)dx  \nonumber \\
&\label{eqn:synth_right_cond}=\igzo K_{1,x}(1,x)y(x,t)dx+M(1)y_x(1,t).
\end{align}

Substituting~\eqref{eqn:synth_left_cond} and~\eqref{eqn:synth_right_cond} into~\eqref{eqn:synth_thm_1} and using the definitions of $R_1$ and $R_2(x)$ produces

\begin{alignat*}{2}
&\frac{d}{dt}V(w(t))=\ip{\mathcal{A}\mathcal{P}y(t)}{y(t)}&&+\ip{y(t)}{\mathcal{A}\mathcal{P}y(t)}  \\
&\leq  \igzo y(x,t) \bigg(T_0(x)_{2,2}y(x,t)&&+\igzx T_1(x,s)y(s,t)ds  \\
&  &&+\igxo T_2(x,s)y(s,t)ds   \bigg) dx,
\end{alignat*} where we have used the fact that $R_1$ and $R_2(x)$ cancel the boundary terms $T_{0_{1,1}}$ and $T_{0_{1,2}}$.
From the Theorem hypotheses,
\[
\left\{-T_{0_{2,2}}- 2 \delta M,-T_1-2 \delta K_1, -T_2-2 \delta K_2\right\} \in \Xi_{d_1,d_2,0}.
\]
Thus we conclude that
\[\frac{d}{dt}V(w(t)) \leq -2 \delta V(w(t)) , \quad t >0.\]
Integrating in time yields
\begin{align*}
 V(w(t)) \leq e^{-2 \delta t} V(w(0)) \Rightarrow & \langle \mathcal{P}y(\cdot,t), y(\cdot,t) \rangle \\
  & \leq  e^{-2 \delta t} \langle w_0, \mathcal{P}^{-1}w_0 \rangle.
\end{align*}

Since $\{M,K_1,K_2\} \in \Xi_{d_1,d_2 \epsilon}$, $\epsilon \|y(\cdot,t)\|^2 \leq \langle \mathcal{P}y(\cdot,t), y(\cdot,t) \rangle$ and thus

\[
\|y(\cdot,t)\|   \leq e^{- \delta t} \sqrt{\frac{\langle w_0, \mathcal{P}^{-1}w_0 \rangle}{\epsilon}}.
\]
Hence,
\begin{align*}
\|w(\cdot,t)\|=\|(\mathcal{P}y)(\cdot,t)\| \leq &  \|\mathcal{P}\|_\mathcal{L} \|y(\cdot,t)\| \\
\le &  \|\mathcal{P}\|_\mathcal{L} e^{- \delta t} \sqrt{\frac{\langle w_0, \mathcal{P}^{-1}w_0 \rangle}{\epsilon}}.
\end{align*}
Which concludes the proof.
 \end{proof}

\subsection{Numerical Results for Full-State Feedback Synthesis}\label{contsynthnum}
\paragraph*{Example 5} In this example, we apply Theorem~\ref{thm:synthesis} to Example 2 from the section on stability analysis. Specifically, System~\eqref{eqn:PDE_cont_form} -~\eqref{eqn:PDE_cont_form_BC} with $a(x)=x^3-x^2+2$ and $b(x)=3x^2-2x$ and $c(x)= -0.5x^3 + 1.3x^2 - 1.5x + 0.7+\lambda$. Table \ref{table_synthesis_1} presents the maximum  $\lambda$, for which a controller can be constructed, as a function of degree $d=d_1=d_2$.
\begin{table}[h]
\begin{center}
    \begin{tabular}{l *{7}{c}}\hline \hline
  & $d=4$ & $5$ & $6$ & $7$ \\ \hline
$\lambda$ &  $15$ & $18$ & $25.9$ & $35$ \\
\end{tabular}
\end{center}
\caption{Maximum $\lambda$ under feedback as a function of polynomial degree, $d=d_1=d_2$ for Example 5 with $\delta=0.1$ and $\epsilon=0.001$.}
\label{table_synthesis_1}
\end{table}The maximum $\lambda$ for which we can construct an exponentially stabilizing controller for is $\lambda=35$, which is $651.1 \%$ increase over the stability margin of $4.66$ which was predicted using finite-difference methods in the previous section. A static controller of the form $u(t)=-kw(1,t)$, $k > 0$, can also be devised using Sturm-Liouville theory \cite[Chapter~5]{egorov1996spectral}. Such a static controller can stabilize the system for $\lambda < 17.58$. The presented methodology can stabilize the system for $\lambda=35$, which is an increase of $99.09 \%$ over $\lambda=17.58$.

Figure \ref{fig:control_control} illustrates the state evolution of the controlled system for $\lambda=35$, $\delta=0.1$ and $\epsilon=0.001$ and the required control effort. Finally, Figure \ref{fig:V_control} illustrates the Lyapunov functional and its time derivative for the controlled system. The initial condition is chosen arbitrarily as 
\begin{equation}\label{eqn:initial}
w_0(x)=e^{-\frac{(x-0.3)^2}{2 (0.07)^2}}-e^{-\frac{(x-0.7)^2}{2 (0.07)^2}}.
\end{equation}

\begin{figure}[h!]
\centering
\subfigure[State evolution]{%
\includegraphics[scale=0.22]{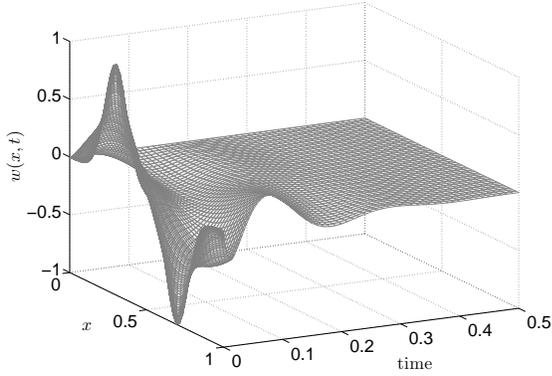}
\label{fig:subfigure1}}
\quad
\subfigure[Control effort $u(t)=(\mathcal{F}w)(t)=(\mathcal{Z}\mathcal{P}^{-1}w)(t)$]{%
\includegraphics[scale=0.22]{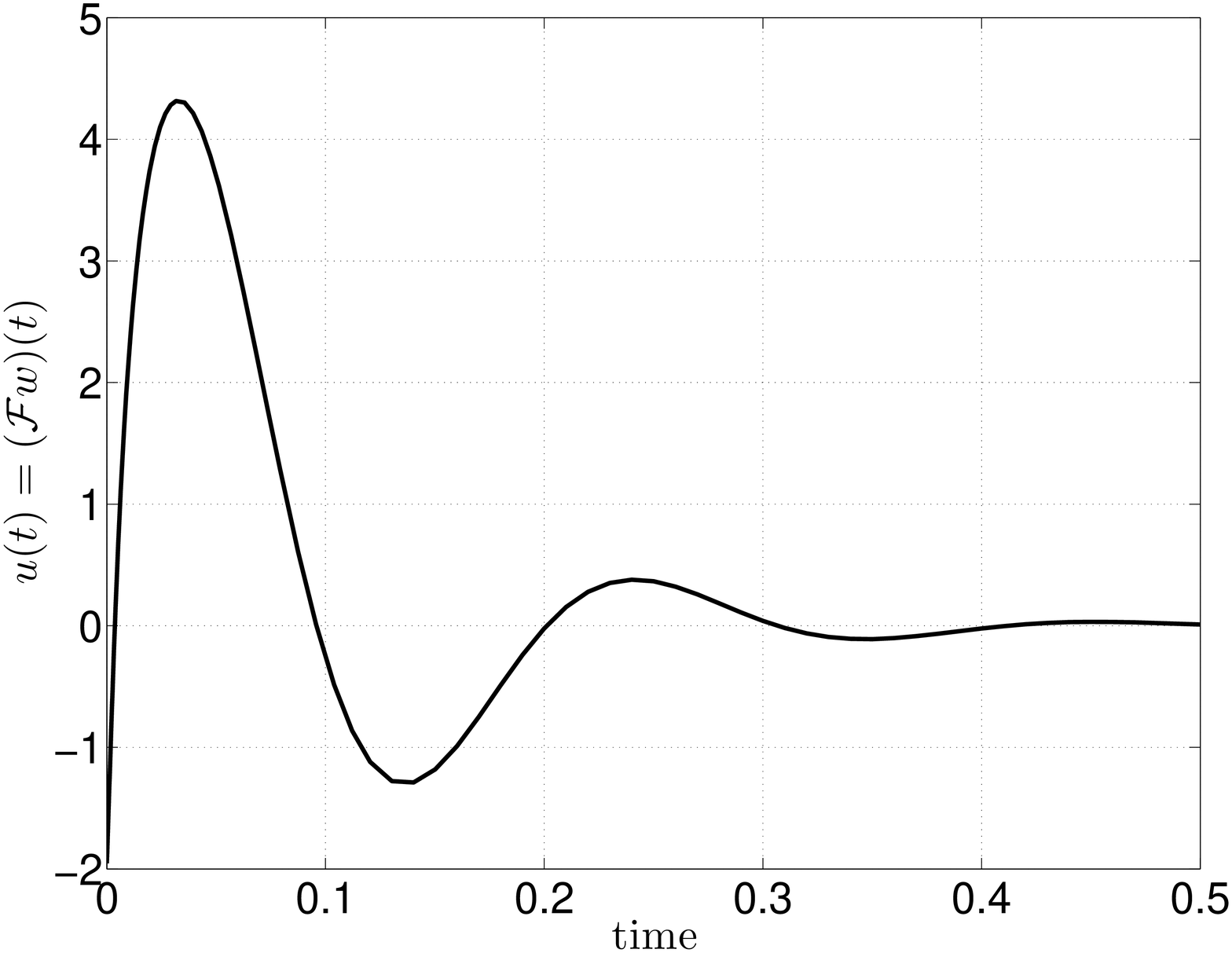}
\label{fig:subfigure2}}
\caption{Evolution of state and input for $a(x)=x^3-x^2+1$ and $b(x)=3x^2-2x$ and $c(x)= -0.5x^3 + 1.3x^2 - 1.5x + 0.7+\lambda$ with $\lambda=35$ and $\delta=0.1$ in Example 5.}
\label{fig:control_control}
\end{figure}

\begin{figure}[h!]
\centering
\subfigure[Illustration of $V(t) \geq \epsilon \|w(\cdot,t)\|^2$.]{%
\includegraphics[scale=0.22]{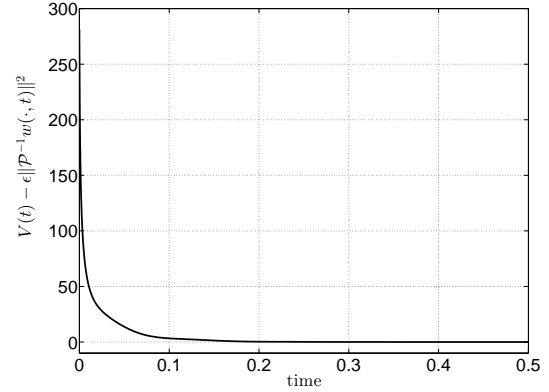}
\label{}}
\quad
\subfigure[Illustration of $\dot{V}(t) \leq -2 \delta V(t) $. ]{%
\includegraphics[scale=0.22]{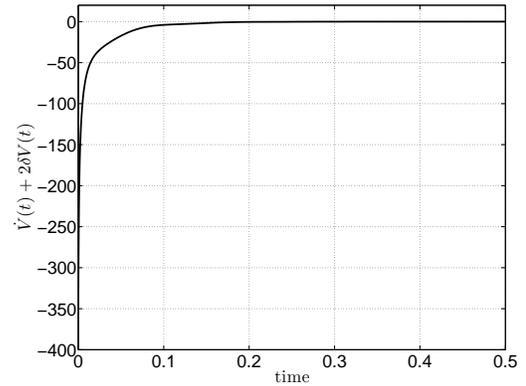}
\label{}}
\caption{Lyapunov functional and its derivative for the controlled system with $\delta = 0.1$ and $\epsilon =0.001$.}
\label{fig:V_control}
\end{figure}

\paragraph*{Example 6} In this example, we apply Theorem~\ref{thm:synthesis} to System~\eqref{eqn:PDE_cont_form} -~\eqref{eqn:PDE_cont_form_BC} with $a(x)=x^3-x^2+2$ and $b(x)=3x^2-2x$ and $c(x)= -0.5x^3 + 1.3x^2 - 1.5x + 6.7$. These values render the system unstable as verified by numerical simulation in Figure \ref{fig:control_auto}.
\begin{figure}[h!]
\centering
\subfigure[State evolution]{%
\includegraphics[scale=0.22]{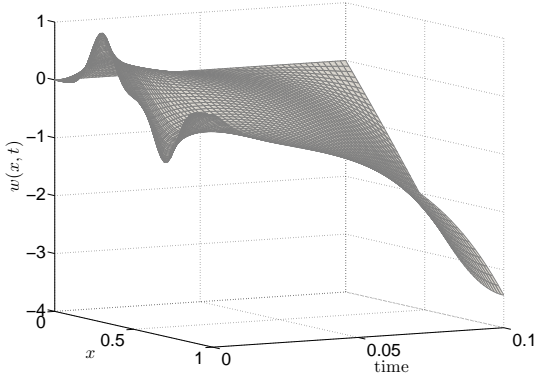}
\label{fig:subfigure1}}
\quad
\subfigure[State norm evolution]{%
\includegraphics[scale=0.22]{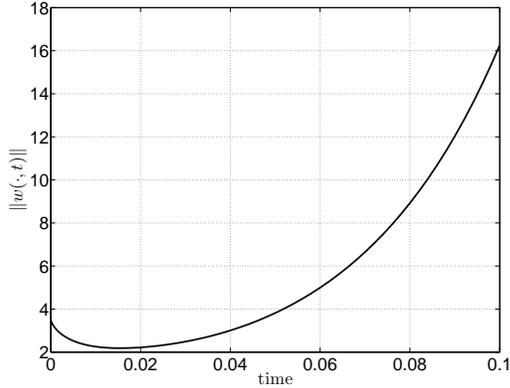}
\label{fig:subfigure2}}
\caption{Evolution of autonomous state for $a(x)=x^3-x^2+2$ and $b(x)=3x^2-2x$ and $c(x)= -0.5x^3 + 1.3x^2 - 1.5x + 6.7$.}
\label{fig:control_auto}
\end{figure}
We wish to find the maximum exponential decay rate $\delta$ for which we can construct a controller. Table \ref{table_synthesis_2} presents the results. 
\begin{table}[h]
\begin{center}
    \begin{tabular}{l *{7}{c}}\hline \hline
  & $d=4$ & $5$ & $6$ & $7$ \\ \hline
$\delta$ &  $1.7$ & $2.9$ & $20.9$ & $22$ \\
\end{tabular}
\end{center}
\caption{Maximum decay rate $\delta$ under feedback as a function of polynomial degree, $d=d_1=d_2$ for Example 5 with $\epsilon=0.001$.}
\label{table_synthesis_2}
\end{table}
As we see, the maximum $\delta$ for which we can construct an exponentially stabilizing controller is $\delta=22 $. This is an increase of $89.98 \%$ over $\delta=11.58$ for which an exponentially stabilizing controller can be constructed using Sturm-Liouville theory.

\paragraph*{Example 7} The presence of the integral kernels $K_1$ and $K_2$ in the Lyapunov functional operator $\mathcal{P}$ necessitates the inclusion of $R_2$ in the control operator $\mathcal{Z}$. As a result, if we wish to use this controller with only an output, instead of the complete state, available for design, an observer is required to be constructed. Thus, it is important to establish the performance improvement gained by the inclusion of $K_1$, $K_2$ and $R_2$. For this purpose, we compare the results obtained in Example 5 to the results obtained for a simple static output feedback based controller which is achieved by setting $K_1=K_2=0$ and $R_2=0$. We apply Theorem~\ref{thm:synthesis}, for $K_1=K_2=0$ and $R_2=0$, on the System considered in Example 5, that is, with $a(x)=x^3-x^2+2$ and $b(x)=3x^2-2x$ and $c(x)= -0.5x^3 + 1.3x^2 - 1.5x + 0.7+\lambda$. Table \ref{table_synthesis_3} presents the maximum  $\lambda$, for which a static controller can be constructed, as a function of degree $d=d_1=d_2$.
\begin{table}[h]
\begin{center}
    \begin{tabular}{l *{7}{c}}\hline \hline
  & $d=4$ & $5$ & $6$ & $7$ \\ \hline
$\lambda$ &  $9.1$ & $9.24$ & $9.24$ & $9.24$ \\
\end{tabular}
\end{center}
\caption{Maximum $\lambda$, for $K_1=K_2=0$ and $R_2=0$, as a function of polynomial degree, $d=d_1=d_2$ for Example 7 with $\delta=0.1$ and $\epsilon=0.001$.}
\label{table_synthesis_3}
\end{table} Upon comparing these results with the ones presented in Table~\ref{table_synthesis_1}, it is evident that the inclusion of $K_1$, $K_2$ and $R_2$ produces much sharper results.

\section{Observer Synthesis}\label{obsynth}
Recall the dynamics of System~\eqref{eqn:PDE_form}:
\begin{equation}
 w_t(x,t)=a(x)w_{xx}(x,t)+b(x)w_x(x,t)+c(x)w(x,t)
\end{equation}
with output $z(t)=\mathcal{C}w(t) = w(1,t)$. Because of the infinite-dimensional nature of PDEs of the Form~\eqref{eqn:PDE_form}, real-time measurement of the state is not possible. For this reason, any realistic approach to control must include an observer and must account for the error dynamics in the closed-loop response. The simplest form of observer for which it is possible to verify closed-loop stability is the Luenberger observer. In our version of the Luenberger observer, the dynamics of the state estimate, $\hat w$ are defined by operator $\mathcal{O}:\lt \rightarrow \lt$ and $O_1 \in \R$ as
\begin{align}
 \wh_t(x,t)=&a(x)\wh_{xx}(x,t)+b(x)\wh_x(x,t)+c(x)\wh(x,t) \nonumber \\
 &\label{eqn:obs_coupled_1}+(\mathcal{O}(\hat{z}(t)-z(t)))(x),
\end{align}
where $\hat z(t)=\mathcal{C}\wh(t)=\wh(1,t)$ with boundary conditions
\begin{align}
 &\label{eqn:obs_coupled_BC_1}\wh(0,t)=0, \qquad \wh_x(1,t)=O_1 (\hat{z}(t)-z(t))+u(t),
\end{align} 
where recall that in feedback $u(t)=\mathcal{F}\wh(t)$ and hence the state itself satisfies
\begin{equation}
 w_t(x,t)\label{eqn:obs_coupled_2}=a(x)w_{xx}(x,t)+b(x)w_x(x,t)+c(x)w(x,t)
\end{equation}
with output $z(t) = w(1,t)$ and boundary conditions \begin{align}
 &\label{eqn:obs_coupled_BC_2}w(0,t)=0, \qquad w_x(1,t)=u(t)=\mathcal{F}\wh(t).
\end{align}

A block-diagram of the coupled dynamics can be found in Figure~\ref{fig:schema}.
\begin{figure}[h!]
\centering
\begin{tikzpicture}[font=\Large,node distance = 1cm, auto,scale=0.6, every node/.style={transform shape}]
 
\node [block, label=System] (system) {\begin{align*} &w_t(x,t)=a(x)w_{xx}(x,t)+b(x)w_x(x,t)+c(x)w(x,t) \end{align*}
\[w(0,t)=0, \quad w_x(1,t)=u(t), \quad z(t)=w(1,t)\]};

\node [block, below=4cm of system] (observer) {\begin{align*} \wh_t(x,t)=& a(x)\wh_{xx}(x,t)+b(x)\wh_x(x,t)+c(x)\wh(x,t) \\
& +(\mathcal{O}(\hat{z}(t)-z(t)))(x) \end{align*}
\[\wh(0,t)=0,  \quad  \wh_x(1,t)=O_1(\hat{z}(t)-z(t))+u(t)\]};

\node [cloud,label={[label distance=-0.6cm]0:$-$},label={[label distance=-0.6cm]270:$+$}, below right=2.60cm and 0.3cm of system](add1) {};

\draw [->] (system) -- (8,0)
node [midway, sloped, above] {$z(t)$};
\draw [->] (7.5,0) |- (add1);

\node [smallblock, left=1cm of add1](L){$\mathcal{O}$};
\draw[->](add1)--(L);
\coordinate [left=2cm of L] (c1);
\coordinate [below=1.2cm of c1] (c2);
\draw [-] (L) to (c1);
\draw[->] (c1) to  (c1 |- observer.north);

\coordinate [below left=-0.5cm and 1cm of observer] (c3);
\draw [-] (observer.west|-c3) -- (c3);
\coordinate [below=1cm of c3](c4);
\draw [-] (c3) to (c4);
\draw [->] (c4) -| (add1)
node [near start, sloped, below] {$\hat{z}(t)$};

\node [smallblock, above left=-0.1cm and 2.5cm of L](U1){$O_1$};
\node [cloud,label={[label distance=-0.6cm]0:$+$},label={[label distance=-0.6cm]180:$+$}, left=2cm of U1](add2) {};
\coordinate [right=0.5cm of L](c5);
\draw [->] (c5) |- (U1)
node [near end, sloped, above] {$\hat{z}(t)-z(t)$};
\draw [->] (U1) to (add2);
\draw [->] (add2) to (observer.north -| add2);

\node [smallblock, below left=0.2cm and 1.5cm of add2](F){$\mathcal{F}$};
\coordinate [above left=-0.5cm and 1cm of observer] (c6);
\draw [->] (c6)|-(F)
node [near end, sloped, above] {$\wh(x,t)$};
\draw [-] (observer.west |- c6) to (c6);
\coordinate [left=1.5cm of add2] (c7);
\draw [->] (F) |- (add2)
node [near end, sloped, above] {$u(t)$};

\coordinate [left=2.5cm of system] (c8);
\draw [-] (c7) -| (c8);
\draw [->] (c8) -- (system.west)
node [midway, above] {$u(t)$};

\coordinate [above left=2.5cm and 1.5cm of observer] (A);
\coordinate [above right=2.5cm and 1.5cm of observer] (B);
\coordinate [below left=1.2cm and 1.5cm of observer] (C);
\coordinate [below right=1.2cm and 1.5cm of observer] (D);
\node [draw=black, dashed, ultra thick,label={[anchor=south]below:Observer Based Controller}, fit= (observer) (A) (B) (C) (D)]{};
\end{tikzpicture}
\caption{Schema representing the coupled dynamics \eqref{eqn:obs_coupled_1}-\eqref{eqn:obs_coupled_BC_2}}
\label{fig:schema}
\end{figure}
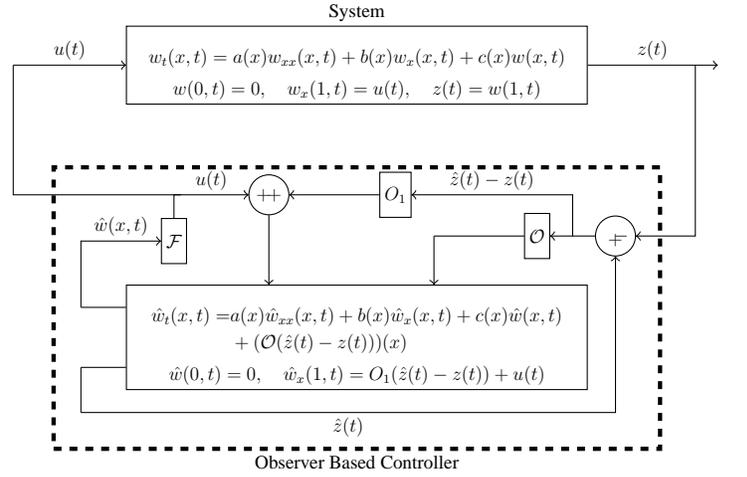

For the coupled dynamics, we consider the following coupled initial conditions
\begin{equation}
w(x,0)=w_0(x) \in H^2(0,1)\quad \text{and }\quad \hat w(x,0)=\wh_0(x) \in H^2(0,1),\label{eqn:obs_init_1}
\end{equation}
where we assume the initial conditions are consistent with the equations as .
\begin{align}
&w_0(0)=0, \quad \wh_0(0)=0, \quad w_{0,x}(1)=\mathcal{F}\wh_0, \quad \text{and} \nonumber \\
&\label{eqn:obs_init_2} \wh_{0,x}(1)=O_1(\wh_0(1)-w_0(1))+\mathcal{F}\wh_0.
\end{align}

In finite-dimensional systems, the Luenberger observer has the property that the eigenvalues of the closed-loop system is the union of the eigenvalues of $\mathcal{A}+\mathcal{LC}$ and the eigenvalues of $\mathcal{A}+\mathcal{BF}$. This implies that stability in closed-loop is equivalent to stability of these two subsystems.

In the following theorem, we prove the analogue of this result for System~\eqref{eqn:PDE_form} in feedback using the Luenberger observer. Our conditions have the form of the following Linear Operator Inequality.
\begin{align}
&\ip{(\mathcal{A}\mathcal{P}+\mathcal{B}\mathcal{Z}) x}{x}+\ip{x}{(\mathcal{A}\mathcal{P}+\mathcal{B}\mathcal{Z}) x} \le -\epsilon \norm{x}^2\\
&\ip{(\mathcal{P}\mathcal{A}+\mathcal{V}\mathcal{C})x}{\mathcal{P}x}\ip{x}{(\mathcal{P}\mathcal{A}+\mathcal{V}\mathcal{C}) x} \le -\epsilon \norm{x}^2
\end{align}

\begin{theorem}\label{observer_synth_colloc}
Suppose there exist
\[
\{M_c,K_{1,c},K_{2,c}\} \in \Omega_{d_1,d_2,\epsilon},\quad \{M_o,K_{1,o},K_{2,o}\} \in \Omega_{d_1,d_2,\epsilon} 
\] and $\epsilon, \delta >0$,
such that
\begin{align*}
& \left\{-T_{0_{2,2}}-2\delta M_c, -T_1-2\delta K_{1,c}, -T_2- 2\delta K_{2,c} \right\} \in \Xi_{d_1,d_2,0}, \\
& \left\{-Q_{0_{2,2}}-2\delta M_o, -Q_1-2\delta K_{1,o}, -Q_2-2\delta K_{2,o} \right\} \in \Xi_{d_1,d_2,0}, \\
 &K_{2,o}(0,x)=0 \quad \text{ and }\quad K_{2,c}(0,x)=0.
\end{align*}
where
\begin{alignat*}{2}
 &(P_c v)(x)= M_c(x)v(x) &&+ \igzx K_{1,c}(x,\xi) v(\xi) d \xi  \\
 &\qquad \qquad \qquad &&+ \igxo K_{2,c}(x,\xi) v(\xi) d \xi,\\
 &(P_o v)(x)= M_o(x)v(x) &&+ \igzx K_{1,o}(x,\xi) v(\xi) d \xi \\
 &\qquad \qquad \qquad &&+ \igxo K_{2,o}(x,\xi) v(\xi) d \xi,\\
&\{T_0,T_1,T_2,T_3,T_4\}=&&\mathcal{N}_\epsilon(M_c,K_{1,c},K_{2,c}),\text{ and }\\
&\{Q_0,Q_1,Q_2,Q_3\}=&&\mathcal{M}_\epsilon(M_o,K_{1,o},K_{2,o}).
\end{alignat*}

Let
\[
\mathcal{F}w:=\mathcal{Z}\mathcal{P}^{-1}_c w \qquad \text{and}\qquad \mathcal{O}w:=\mathcal{P}_o^{-1}\mathcal{V}w
\] 
where
\begin{align*}
&(\mathcal{Z}y)=R_1y(1)+\igzo R_2(x)y(x)dx, \quad R_1=-\frac{T_{0_{1,1}}}{2 a(1)},\\
&  R_2=-\frac{T_{0_{1,2}}}{a(1)},\quad O_1=\frac{1}{2a(1)M_o(1)}\left(a_x(1)M(1)+a(1)M_{o,x}(1) \right),
\end{align*} and
\begin{align*}
&\mathcal{V}r=\left[\left(a_x(1)-O_1 a(1)-b(1) \right)K_{1,o}(1,x)+a(1)K_{1,o,x}(1,x) \right]r.
\end{align*}
Then, for initial conditions $w_0$ and $\wh_0$ given in \eqref{eqn:obs_init_1}-\eqref{eqn:obs_init_2}, there exists a constant $\omega>0$ such that any solution $\{w,\hat w\}$ of~\eqref{eqn:obs_coupled_1}-\eqref{eqn:obs_coupled_BC_2} satisfies
\[
  \left\|\bmat{w(t) \\ \wh(t) }\right\|_{\lt} \leq \omega e^{-\delta t} \left\|\bmat{w_0 \\ \wh_0 }\right\|_{\lt}.
\]

\end{theorem}

\begin{proof}
We begin by defining the state estimation error $e(x,t)=\wh(x,t)-w(x,t)$, the dynamics of which are given by
 \begin{equation}
 \label{obs_error_eqn}
 e_t(x,t)=a(x)e_{xx}(x,t)+b(x)e_x(x,t)+c(x)e(x,t)+(\mathcal{O}e(1,t))(x)
\end{equation} with boundary conditions
\begin{equation}
 \label{obs_error_eqn_BCs}
 e(0,t)=0, \qquad e_x(1,t) = O_1 e(1,t).
\end{equation}
For the error system, we define the following Lyapunov functional
\[
V(e(t))= \langle e(t),\mathcal{P}_o e(t) \rangle.
\]
Taking the time derivative yields
 \begin{align*}
\frac{d}{dt}V(t)=& \langle e_t(t),\mathcal{P}_o e(t) \rangle + \langle e(t),\mathcal{P}_o e_t(t) \rangle \\
=& \langle \mathcal{A} e(t),\mathcal{P}_o e(t) \rangle + \langle e(t),\mathcal{P}_o\mathcal{A}e(t) \rangle + 2 \ip{\mathcal{O}e(1,t)}{\mathcal{P}_o e(t)}.
\end{align*}
Let $\{Q_0,Q_1,Q_2,Q_3\}=\mathcal{M}_\epsilon(M_o,K_{1,o},K_{2,o})$ then $K_{2,o}(0,x)=0$ implies $Q_3(x)=0$, and hence Corollary~\ref{cor:primal2} and $e_x(1,t)=O_1(\hat z(t)-z(t))=O_1 e(1,t)$ imply
 \begin{align}
 &\frac{d}{dt}V(t) \nonumber \\
 & \leq   \ip{\bmat{e(1,t) \\ e(\cdot,t)}}{ \mathcal{Q} \bmat{e(1,t) \\ e(\cdot,t)}}_{\R \times \lt} \nonumber \\
 & \qquad + 2 \igzo (\mathcal{O}e(1,t))(x) (\mathcal{P}_o e)(x,t)dx + 2 O_1 a(1)M_o(1)e^2(1,t) \nonumber \\
 &\label{eqn:Vdot_obs1}\qquad +  2e(1,t) O_1 \igzo a(1)K_{1,o}(1,x) e(x,t)dx .
 \end{align}
 where
\begin{align*}
(\mathcal{Q}y)(s) := &Q_0(s)\bmat{y(1) \\ y(s)} +   \igzs  \bmat{0 & 0 \\ 0 & Q_1(s,t)} \bmat{y(1) \\ y(t)} dt \\
&+ \igso  \bmat{0 & 0 \\ 0 & Q_2(s,t)} \bmat{y(1) \\ y(t)} dt.
\end{align*}
Now,
 \[
 \igzo (\mathcal{O}e(1,t))(x) (\mathcal{P}_oe)(x,t)dx = \igzo (\mathcal{P}_o \mathcal{O}e(1,t))(x) e(x,t)dx.
 \]
and  $\mathcal{V}=\mathcal{P}_o\mathcal{O}$ implies

 \begin{align}
 &\igzo (\mathcal{O}e(1,t))(x) (\mathcal{P}_0e)(x,t)dx =  \igzo (\mathcal{V}e(1,t))(x) e(x,t)dx \nonumber \\
&= e(1,t) \igzo \left( \vphantom{\pfs}(a_x(1)-O_1a(1)-b(1))K_{1,o}(1,x) \right. \nonumber \\
&\label{eqn:Vdot_obs_inter}\left. \vphantom{\pfs}\qquad \qquad  \qquad   +a(1)K_{1,o,x}(1,x) \right) e(x,t)dx.
\end{align} Substituting Equation~\eqref{eqn:Vdot_obs_inter} into \eqref{eqn:Vdot_obs1}, yields
\begin{alignat*}{2}
 &\frac{d}{dt}V(t) \\
 &\leq \igzo e(x,t) \bigg(Q_0(x)_{2,2}e(x,t)&&+\igzx Q_1(x,s)e(s,t)ds  \\
 & &&+\igxo Q_2(x,s)e(s,t)ds    \bigg)dx,
\end{alignat*} where, the boundary terms have been canceled due to $\mathcal{V}$ and $O_1$.

Since we have
\[
\left\{-Q_{0_{2,2}}-2\delta M_o, -Q_1-2\delta K_{1,o}, -Q_2-2\delta K_{2,o} \right\} \in \Xi_{d_1,d_2,0},
\]
we conclude that
\[
 \frac{d}{dt}V(t) \leq -2  \delta V(t), \quad t>0.
\]
Since $\{M_o,K_{1,o},K_{2,o}\} \in \Xi_{d_1,d_2,\epsilon}$, we have
 \[
 \|e(\cdot,t)\| \leq e^{-\delta t} \sqrt{\frac{\ip{e_0}{\mathcal{P}_o e_0}}{\epsilon}}, \quad t>0.
 \]
Now, since the state satisfies
\begin{equation}
 w_t(x,t)\label{eqn:obs_coupled_2}=a(x)w_{xx}(x,t)+b(x)w_x(x,t)+c(x)w(x,t)
\end{equation}
with  $w(0,t)=0$ and $w_x(1,t)=\mathcal{F}\wh(t)=\mathcal{F} w(t) + \mathcal{F}e(t)$ then by applying $\mathcal{P}_c$ to Theorem~\ref{thm:synthesis}, we conclude exponential stability of the coupled system. which implies the existence of an $\omega > 0$ such that
\[
  \left\|\bmat{w(t) \\ \wh(t) }\right\|_{\lt^2} \leq \omega e^{-\delta t} \left\|\bmat{w_0 \\ \wh_0 }\right\|_{\lt^2}.
\]
\end{proof}
 Note that in this theorem we have chosen a common positivity margin $\epsilon>0$ and exponential decay rate $\delta>0$ for the controller and observer synthesis conditions. In practice, it is customary to choose a faster decay rate for the observer than the controller. In this case, the conditions should be modified accordingly.

\subsection{Observer Synthesis Numerical Results}

\paragraph{Example 8} In this final section, we perform numerical experiments on the same example presented in Section~\ref{contsynthnum}. Specifically, we apply Theorem \ref{observer_synth_colloc} to System~\eqref{eqn:obs_coupled_1}-\eqref{eqn:obs_coupled_BC_2} with $a(x)=x^3-x^2+2$, $b(x)=3x^2-2x$, $c(x)= -0.5x^3 + 1.3x^2 - 1.5x + 0.7+\lambda$. The results presented here are simulations obtained using the observer based controller $u(t)=(\mathcal{F}\wh)(t)$ given by the conditions of Theorem~\ref{observer_synth_colloc} and obtained using the operator inversion technique described in Theorem~\ref{thm:inv_op}.

Table \ref{table_obs_2} presents the maximum  $\lambda$ for which an observer can be constructed using $\epsilon=0.001$ and $\delta=0.1$ as a function of degree $d=d_1=d_2$. Figure \ref{fig:obs} illustrates the evolution of the trajectory of the state estimate $\wh(x,t)$, system state $w(x,t)$ and the error state $e(x,t)=\wh(x,t)-w(x,t)$ for $\delta=0.1$. Finally, Figure~\ref{fig:V_observer} illustrates the Lyapunov functional defined in the proof of Theorem~\ref{observer_synth_colloc} for the error dynamics. The initial condition $w_0(x)$ is given in Equation~\eqref{eqn:initial} and for the observer we choose $\wh_0(x)=0$.

\begin{table}[h!]
\begin{center}
    \begin{tabular}{l *{7}{c}}\hline \hline
  & $d=4$ & $5$ & $6$ & $7$ \\ \hline
$\lambda$ &  $15$ & $18$ & $25.9$ & $35$ \\
\end{tabular}
\end{center}
\caption{Maximum $\lambda$ of the error system as a function of $d=d_1=d_2$ for Numerical Example 8.}
\label{table_obs_2}
\end{table}

\begin{figure}[ht]
\centering
\subfigure[Observer state evolution]{%
\includegraphics[scale=0.22]{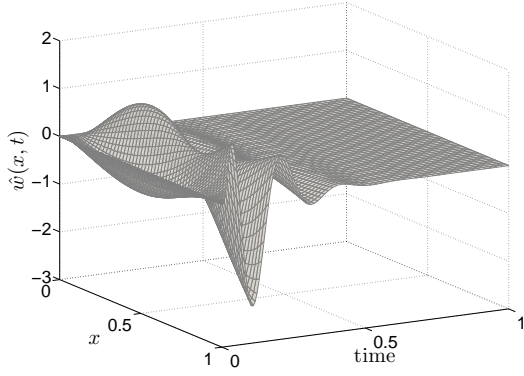}
\label{fig:subfigure1}}
\quad
\subfigure[System state evolution]{%
\includegraphics[scale=0.22]{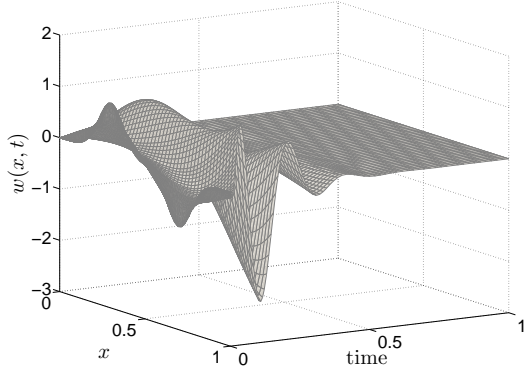}
\label{fig:subfigure1}}
\quad
\subfigure[Error in the estimate of the state]{%
\includegraphics[scale=0.22]{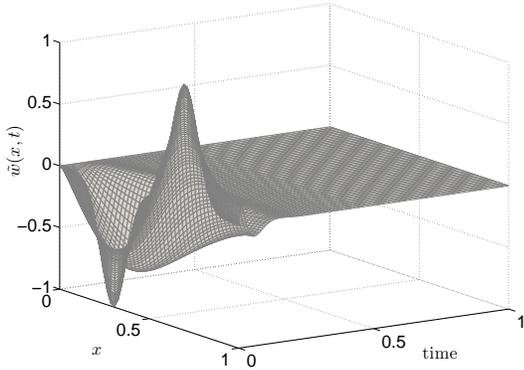}
\label{fig:subfigure1}}
\caption{Evolution of the observer state $\wh(x,t)$, the system state $w(x,t)$ and the error state $e(x,t)$.}
\label{fig:obs}
\end{figure}

\begin{figure}[ht]
\centering
\subfigure[Illustration of $V(t) \geq \epsilon \|\wt(\cdot,t)\|^2$.]{%
\includegraphics[scale=0.22]{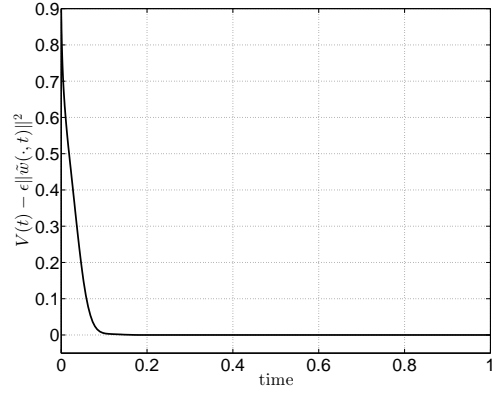}
\label{}}
\caption{Lyapunov functional for the error system with $\delta = 0.1$ and $ \epsilon =0.001$.}
\label{fig:V_observer}
\end{figure}

\section{Conclusion}

 In this paper, we have developed a algorithmic approach to the design of observer-based controllers for a general class of scalar parabolic partial differential equations using point measurements and feedback at the boundary. The results use the sum-of-squares methodology to parameterize a convex set of positive operators. In this way we cast the problem of controller synthesis in the framework of convex optimization - a class of optimization problems for which we have efficient numerical algorithms. Furthermore, we have applied our results to a difficult numerical example in order to demonstrate that our results are practical and effective. The reader is invited to contemplate natural extensions of this work including the development of methods for control of coupled partial-differential equations. We also speculate that the conditions as stated are conservative and may be improved through a generalization of the Wirtinger inequality, or some other method for relating state parameters $w, w_s, w_{ss}, w(1)$, etc. Additional possibilities include application to other classes of PDE system.

%

\appendix\label{sec:appendix}
First, recall the variation  of Wirtinger's Inequality.
\begin{lemma}[\cite{hardy1952inequalities},\cite{krstic2008boundary}]
\label{lem:wirtinger}
 let $z \in H^2(0,1)$ be a scalar function. Then
  \[\int_0^1 (z(s))^2ds \leq (z(0))^2+ \frac{4}{\pi^2} \int_0^1 (z_s(s))^2 ds.\]
\end{lemma}
Now recall the definition of $\mathcal{M}_\epsilon$.
\begin{definition}
We say $\{Q_0,Q_1,Q_2,Q_3\}=\mathcal{M}_\epsilon(M,K_1,K_2)$ if the following hold
\begin{align}
 &Q_0(s)_{1,1}=\left[\left(b(1)-a_s(1) \right)M(1)-a(1)M_s(1) \right],\\
 &Q_0(s)_{1,2}=Q_0(s)_{2,1} \nonumber \\
 &=\left[\left(b(1)-a_s(1) \right)K_1(1,s)-a(1)K_{1,s}(1,s) \right], \\
&Q_0(s)_{2,2}=\pfs \left[ \pfs \left[ a(s)M(s)\right]-b(s)M(s) \right]+2M(s)c(s) \nonumber \\
&   +  \left[\pfs \left[2a(s)\left(K_1(s,t)-K_2(s,t) \right) \right] \right]_{t=s}-\frac{\pi^2}{2}\alpha \epsilon ,\\
&Q_1(s,t) \nonumber \\
&=\left(\pfs \left[\pfs \left[a(s)K_1(s,t)\right]-b(s)K_1(s,t) \right]+c(s)K_1(s,t) \right) \nonumber \\
&+\left(\pft \left[\pft \left[a(t)K_1(s,t)\right]-b(t)K_1(s,t) \right]+c(t)K_1(s,t) \right),  \\ 
&Q_2(s,t)=Q_1(t,s)\nonumber \\
&=\left(\pfs \left[\pfs \left[a(s)K_2(s,t)\right]-b(s)K_2(s,t) \right]+c(s)K_2(s,t) \right) \nonumber \\
&+\left(\pft \left[\pft \left[a(t)K_2(s,t)\right]-b(t)K_2(s,t) \right]+c(t)K_2(s,t) \right)\text{ and }\\
&Q_3(s)=-2a(0)K_2(0,s),
\end{align} where $ K_{1,s}(1,s)=\left[K_{1,s}(s,t)|_{s=1} \right]_{t=s}$.
\end{definition}

\begin{lemma}\label{lem:primal}
Suppose we are given  $\{M,K_1,K_2\} \in \Xi_{d_1,d_2,\epsilon}$ and $\{Q_0,Q_1,Q_2,Q_3\}=\mathcal{M}_\epsilon(M,K_1,K_2)$. Then, for $\mathcal{A}$ as defined in Equation~\eqref{eqn:Aoperator} and $\mathcal{P}$ as defined in Equation~\eqref{eqn:Poperator}, we have that
\begin{align*}
\ip{\mathcal{A}w}{\mathcal{P}w}+\ip{w}{\mathcal{P}\mathcal{A}w} \le &  \ip{\bmat{w(1)\\w}}{ \mathcal{Q} \bmat{w(1)\\w}}_{\R \times L_2} \\
& + \igzo w_s(0)Q_3(s)w(s)ds
\end{align*}
for any $w \in \mathcal{D}_0$ where $\mathcal{D}_0$ is defined in Equation~\eqref{state_space_stab} and where $\mathcal{Q}$ is defined as
\begin{align*}
(\mathcal{Q}y)(s) :=  Q_0(s)\bmat{y(1) \\ y(s)} &+   \igzs  \bmat{0 & 0 \\ 0 & Q_1(s,t)} \bmat{y(1) \\ y(t)} dt \\
&+ \igso  \bmat{0 & 0 \\ 0 & Q_2(s,t)} \bmat{y(1) \\ y(t)} dt.
\end{align*}
\end{lemma}

\begin{proof}
We begin by considering the following decomposition
\begin{align}
&\ip{\mathcal{A}w}{\mathcal{P}w}+\ip{w}{\mathcal{P}\mathcal{A}w}  \nonumber \\
&=2\igzo \left(a(s)w_{ss}(s)+b(s)w_s(s)+c(s)w(s) \right)(\mathcal{P}w)(s)ds \nonumber \\
&\label{Gamma_eqn} = 2\left(\Gamma_1+\Gamma_2+\Gamma_3+\Gamma_4+\Gamma_5\right),
\end{align}
 where
 \begin{align*}
 &\Gamma_1 = \igzo w_{ss}(s)a(s)M(s)w(s)ds, \\
 &\Gamma_2 = \igzo w_s(s)b(s)M(s)w(s)ds, \\
 &\Gamma_3 \\
 &= \igzo w_{ss}(s)a(s) \left(\igzs K_1(s,t)w(t) dt + \igso K_2(s,t)w(t)d t  \right)ds, \\
 &\Gamma_4 = \igzo w_s(s)b(s) \left(\igzs K_1(s,t)w(t) d t + \igso K_2(s,t)w(t)d t  \right)ds
 \end{align*}  and
 \begin{align*}
 \Gamma_5 =  \igzo w(s)^2 M(s) c(s) ds &+ \igzo \igzs w(s) c(s) K_1(s,t)w(t) d t ds \\
 & + \igzo \igso w(s) c(s) K_2(s,t)w(t) dt ds.
 \end{align*}
 Applying integration by parts and using the boundary condition $w(0)=0$ yields
 \begin{align*}
 \Gamma_1 =& - \igzo w_s(s)^2 a(s)M(s)ds + \igzo w(s)^2\left( \frac{1}{2}\frac{\partial^2}{\partial s^2} \left[ a(s)M(s)\right] \right) ds \\
 &- w(1)^2 \left(\frac{1}{2}\left(a_s(1)M(1)+ a(1)M_s(1) \right) \right) \\
 &+ w_s(1)a(1)M(1) w(1).
 \end{align*}

 Since $a(s) \geq \alpha$ and $\{M,K_1,K_2\} \in \Xi_{d_1,d_2,\epsilon}$, we have $a(s)M(s) \geq \alpha \epsilon$. Thus, by application of the Wirtinger Inequality and boundary condition $w(0)=0$, we have
 \[
 - \igzo w_s(s)^2 a(s)M(s) ds \leq -\frac{\pi^2}{4}\alpha \epsilon \igzo w(s)^2 ds.
 \]
We conclude that
 \begin{align}
 \Gamma_1 \leq & \igzo w(s)^2\left[\left( \frac{1}{2}\frac{\partial^2}{\partial s^2} \left[ a(s)M(s)\right] \right)-\frac{\pi^2}{4} \alpha \epsilon\right] dx \nonumber \\
 &  -w(1)^2 \left(\frac{1}{2}\left(a_s(1)M(1)+ a(1)M_s(1) \right) \right) \nonumber \\
 &\label{Gamma_1} +w_s(1)w(1) a(1)M(1).
 \end{align}
Through integration by parts and application of boundary conditions, we also obtain
 \begin{align}
 &\Gamma_2 \nonumber \\
 &\label{Gamma_2}=  - \igzo w(s)^2 \left( \frac{1}{2} \pfs \left[ b(s)M(s)\right] \right)ds + (w(1))^2 \left( \frac{1}{2}b(1)M(1) \right).
 \end{align}
Now, note that for $(M,K_1,K_2) \in \Xi_{d_1,d_2,\epsilon}$, we have $K_1(x,y)=K_2(y,x)$. Exploiting this property, we find 
\begin{align}
 \Gamma_3 =& \igzo w(s)^2\left( \left[ \pfs \left[a(s)(K_1(s,t)-K_2(s,t)) \right] \right]_{t=s} \right) ds \nonumber \\
 & +  \igzo \igzs w(s) \left( \frac{\partial^2}{\partial s^2}\left[a(s)K_1(s,t) \right] \right) w(t) d t ds \nonumber \\
 &+ \igzo \igso w(s) \left( \frac{\partial^2}{\partial s^2}\left[a(s)K_2(s,t) \right] \right) w(t) d t ds \nonumber \\
 &  - w(1) \igzo \left(a_s(1)K_1(1,s)+a(1)K_{1,s}(1,s) \right) w(s) ds \nonumber \\
 &   + w_s(1) \igzo a(1)K_1(1,s) w(s)ds. \nonumber
 \end{align} We can re-write the previous expression as
 \begin{align}
 \Gamma_3 =& \igzo w(s)^2\left( \left[ \pfs \left[a(s)(K_1(s,t)-K_2(s,t)) \right] \right]_{t=s} \right) ds \nonumber \\
 & +  \igzo \igzs w(s) \left(\hlf \frac{\partial^2}{\partial s^2}\left[a(s)K_1(s,t) \right] \right) w(t) d t ds \nonumber \\
 &+ \igzo \igso w(s) \left(\hlf \frac{\partial^2}{\partial s^2}\left[a(s)K_2(s,t) \right] \right) w(t) d t ds \nonumber \\
 &+  \igzo \igzs w(s) \left(\hlf \frac{\partial^2}{\partial s^2}\left[a(s)K_1(s,t) \right] \right) w(t) d t ds \nonumber \\
 &+ \igzo \igso w(s) \left(\hlf \frac{\partial^2}{\partial s^2}\left[a(s)K_2(s,t) \right] \right) w(t) d t ds \nonumber \\
 &  - w(1) \igzo \left(a_s(1)K_1(1,s)+a(1)K_{1,s}(1,s) \right) w(s) ds \nonumber \\
 &   + w_s(1) \igzo a(1)K_1(1,s) w(s)ds. \nonumber
 \end{align} Changing the order of integration in the last two double integrals and switching the variables $s$ and $t$,
 \begin{align}
 \Gamma_3 =& \igzo w(s)^2\left( \left[ \pfs \left[a(s)(K_1(s,t)-K_2(s,t)) \right] \right]_{t=s} \right) ds \nonumber \\
 & +  \igzo \igzs w(s) \left(\hlf \frac{\partial^2}{\partial s^2}\left[a(s)K_1(s,t) \right] \right. \nonumber \\
 &\left. \qquad \qquad +\hlf \frac{\partial^2}{\partial t^2}\left[a(t)K_1(s,t) \right] \right) w(t) d t ds \nonumber \\
 &+ \igzo \igso w(s) \left(\hlf \frac{\partial^2}{\partial s^2}\left[a(s)K_2(s,t) \right] \right. \nonumber \\
 &\left. \qquad \qquad +\hlf \frac{\partial^2}{\partial t^2}\left[a(t)K_2(s,t) \right] \right) w(t) d t ds \nonumber \\
 &  - w(1) \igzo \left(a_s(1)K_1(1,s)+a(1)K_{1,s}(1,s) \right) w(s) ds \nonumber \\
 &\label{Gamma_3}   + w_s(1) \igzo a(1)K_1(1,s) w(s)ds.
 \end{align}
 
 Similarly,
 \begin{align}
 \Gamma_4 =&- \igzo \igzs w(s) \left(\hlf \pfs \left[b(s)K_1(s,t)\right] \right. \nonumber \\
 &\left. \qquad \qquad + \hlf \frac{\partial}{\partial t}\left[ b(t)K_1(s,t)\right] \right) w(t) d t ds  \nonumber \\
 & - \igzo \igso w(s) \left(\hlf \pfs \left[b(s)K_2(s,t) \right] \right. \nonumber \\
 & \left. \qquad \qquad + \hlf \frac{\partial}{\partial t} \left[b(s)K_2(s,t) \right] \right) w(t) d t ds \nonumber \\
 & \label{Gamma_4} + w(1) \igzo b(1)K_1(1,s)w(s)ds.
 \end{align} 
Finally, employing a change of order of integration produces
 \begin{align}
 \Gamma_5 = & \igzo w(s)^2 M(s) c(s) ds \nonumber \\
 &+ \igzo \igzs w(s) \left(\hlf \left[c(s)+c(t) \right] K_1(s,t) \right) w(t) d t ds \nonumber \\
 &\label{Gamma_5} + \igzo \igso w(s) \left( \hlf \left[c(s)+c(t)\right] K_2(s,t)\right)w(t) dt ds.
 \end{align}
 Substituting \eqref{Gamma_1}-\eqref{Gamma_5} into \eqref{Gamma_eqn} gives us
 \begin{align}
 &\ip{\mathcal{A}w}{\mathcal{P}w}+\ip{w}{\mathcal{P}\mathcal{A}w} \nonumber \\
 &\leq    \ip{\bmat{w(1)\\w}}{ \mathcal{Q} \bmat{w(1)\\w}}_{\R \times L_2} +w_s(0) \igzo Q_3(s)w(s)ds \nonumber \\
 &\label{eqn:primal_setup}\quad +2 w_s(1) \left( a(1)M(1)w(1)+ \igzo  a(1)K_1(1,s)w(s)ds \right) .
 \end{align} Since $w \in \mathcal{D}_0$, $w_s(1)=0$. This gives us the desired result.
 \end{proof}

\begin{corollary}\label{cor:primal2}
Suppose we are given  $\{M,K_1,K_2\} \in \Xi_{d_1,d_2,\epsilon}$ and $\{Q_0,Q_1,Q_2,Q_3\}=\mathcal{M}_\epsilon(M,K_1,K_2)$. Then, for $\mathcal{A}$ as defined in Equation~\eqref{eqn:Aoperator} and $\mathcal{P}$ as defined in Equation~\eqref{eqn:Poperator}, we have that
 \begin{align}
 &\ip{\mathcal{A}w}{\mathcal{P}w}+\ip{w}{\mathcal{P}\mathcal{A}w} \nonumber \\
 &\leq    \ip{\bmat{w(1)\\w}}{ \mathcal{Q} \bmat{w(1)\\w}}_{\R \times L_2} +w_s(0) \igzo Q_3(s)w(s)ds \nonumber \\
 &\quad +2 w_s(1) \left( a(1)M(1)w(1)+ \igzo  a(1)K_1(1,s)w(s)ds \right) .
 \end{align}
for any $w \in H^2(0,1)$ with $w(0)=0$ where $\mathcal{Q}$ is defined as
\begin{align*}
(\mathcal{Q}y)(s) := Q_0(s)\bmat{y(1) \\ y(s)} &+   \igzs  \bmat{0 & 0 \\ 0 & Q_1(s,t)} \bmat{y(1) \\ y(t)} dt \\
&+ \igso  \bmat{0 & 0 \\ 0 & Q_2(s,t)} \bmat{y(1) \\ y(t)} dt.
\end{align*}
\end{corollary}
\begin{proof}
Omit the last line in the proof of Lemma~\ref{cor:primal2}.
\end{proof}

The following lemma gives a result which is dual to Lemma~\ref{lem:dual}.

\begin{definition}
We say $\{T_0,T_1,T_2,T_3,T_4\}=\mathcal{N}_\epsilon(M,K_1,K_2)$ if the following hold
\begin{align}
T_0(s)_{1,1}=& \left[-a(1)M_s(1)+(b(1)-a_s(1))M(1) \right], \\
T_0(s)_{1,2}=&T_0(s)_{2,1}=-a(1)K_{1,s}(1,s), \\
T_0(s)_{2,2}=& \left[(a_{ss}(s)-b_s(s))M(s)+b(s)M_s(s)\right]+2 M(s)c(s) \nonumber \\
 & +a(s)\left[M_{ss}(s)+2\pfs \left[K_1(s,t)-K_2(s,t) \right] \right]_{t=s} \nonumber \\
 & -\frac{\pi^2}{2}\alpha \epsilon ,\\
T_1(s,t)=& a(s)K_{1,ss}(s,t)+b(s)K_{1,s}(s,t)+c(s)K_1(s,t) \nonumber \\
&+a(t)K_{1,tt}(s,t)+b(t)K_{1,t}(s,t)+c(t)K_1(s,t), \\
T_2(s,t)=& a(s)K_{2,ss}(s,t)+b(s)K_{2,s}(s,t)+c(s)K_2(s,t) \nonumber \\
&+a(t)K_{2,tt}(s,t)+b(t)K_{2,t}(s,t)+c(t)K_2(s,t),\\
T_3=&a_x(0)M(0)-a(0)M_x(0)-b(0)M(0)+\frac{\pi^2}{2}\alpha \epsilon \text{ and }\\
T_4=&-2 a(0)M(0).
\end{align}
\end{definition}

\begin{lemma}\label{lem:dual}
Suppose $\{M,K_1,K_2\} \in \Xi_{d_1,d_2,\epsilon}$ and $\{T_0,T_1,T_2,T_3,T_4\}=\mathcal{N}_\epsilon(M,K_1,K_2)$. Then, for $\mathcal{A}$ as defined in Equation~\eqref{eqn:Aoperator} and $\mathcal{P}$ as defined in Equation~\eqref{eqn:Poperator}, we have that
\begin{align*}
&\ip{\mathcal{A}\mathcal{P}w}{w}+\ip{w}{\mathcal{A}\mathcal{P}w} \\
&\le  \ip{\bmat{w(1)\\w}}{ \mathcal{T} \bmat{w(1)\\w}}_{\R \times L_2} +w(0)\left(T_3 w(0) + T_4 w_s(0) \right).
\end{align*}
for any $w \in \mathcal{P}^{-1}\mathcal{D}_0$ where $\mathcal{D}_0$ is defined in Equation~\eqref{state_space_stab} and
\begin{align*}
(\mathcal{T}y)(s) := T_0(s)\bmat{y(1) \\ y(s)} &+   \igzs  \bmat{0 & 0 \\ 0 & T_1(s,t)} \bmat{y(1) \\ y(t)} dt \\
& + \igso  \bmat{0 & 0 \\ 0 & T_2(s,t)} \bmat{y(1) \\ y(t)} dt.
\end{align*}
\end{lemma}

\begin{proof}
We begin by considering the following decomposition
\begin{align}
&\ip{\mathcal{A}\mathcal{P}w}{w}+\ip{w}{\mathcal{A}\mathcal{P}w} \nonumber \\
&= 2 \igzo \left(a(s)\frac{\partial^2}{\partial s^2}\left[(\mathcal{P}w)(s) \right]+b(s)\pfs \left[(\mathcal{P}w)(s) \right] \right. \nonumber \\
&\quad \left. \vphantom{\pfs} + c(s)(\mathcal{P}w)(s)  \right)w(s) ds \label{eqn:dual_gamma}= 2 \left(\Gamma_1+\Gamma_2+\Gamma_3+\Gamma_4+\Gamma_5  \right),
\end{align} where
\begin{align*}
&\Gamma_1 = \igzo w(s)a(s) \frac{\partial^2}{\partial s^2} \left[M(s)w(s) \right]ds,\\
&\Gamma_2 = \igzo w(s)b(s) \pfs \left[M(s)w(s) \right]ds,\\
&\Gamma_3 \\
&= \igzo w(s)a(s)\frac{\partial^2}{\partial s^2} \left[\igzs K_1(s,t)w(t)dt + \igso K_2(s,t)w(t)dt \right]ds,\\
&\Gamma_4 \\
&=\igzo w(s)b(s) \pfs \left[\igzs K_1(s,t)w(t)dt + \igso K_2(s,t)w(t)dt \right]ds
\end{align*} and
\begin{align*}
\Gamma_5 = &\igzo w(s)^2 M(s)c(s)ds  + \igzo \igzs w(s)c(s)K_1(s,t)w(t)dtds \\
&+ \igzo \igso w(s)c(s)K_2(s,t)w(t)dtds.
\end{align*}

Applying integration by parts,
\begin{align*}
\Gamma_1 =& -\igzo w_s(s)^2 a(s)M(s)ds \\
&+ \frac{1}{2}\igzo w(s)^2 \left[a_{ss}(s)M(s)+a(s)M_{ss}(s) \right]ds \\
&+ \frac{1}{2}w(1)^2 \left[a(1)M_s(1)-a_s(1)M(1) \right]+w(1)a(1)M(1)w_s(1) \\
&+ \frac{1}{2}w(0)^2 \left[a_s(0)M(0)-a(0)M_s(0) \right] \\
&-w(0)a(0)M(0)w_s(0).
\end{align*} Since $a(s)M(s) \geq \alpha \epsilon$, applying Lemma~\ref{lem:wirtinger} yields
\begin{align*}
&-\igzo w_s(s)^2 a(s)M(s)ds \leq -\frac{\pi^2}{4}\alpha \epsilon \igzo w(s)^2 ds+\frac{\pi^2}{4}\alpha \epsilon w(0)^2. 
\end{align*} Thus
\begin{align}
\Gamma_1 \leq & \igzo w(s)^2 \left( \frac{1}{2}\left[a_{ss}(s)M(s)+a(s)M_{ss}(s) \right]-\frac{\pi^2}{4}\alpha \epsilon \right)ds  \nonumber \\
&+ \frac{1}{2}w(1)^2 \left[a(1)M_s(1)-a_s(1)M(1) \right] \nonumber \\
&+ w(0)^2 \left( \frac{1}{2}\left[a_s(0)M(0)-a(0)M_s(0) \right]+\frac{\pi^2}{4} \alpha \epsilon \right) \nonumber \\
&\label{eqn:dual_gamma_1}+w(1)a(1)M(1)w_s(1)-w(0)a(0)M(0)w_s(0).
\end{align}

Similarly
\begin{align}
\Gamma_2 = & \frac{1}{2}\igzo w(s)^2 \left[b(s)M_s(s)-b_s(s)M(s) \right]ds \nonumber \\
&\label{eqn:dual_gamma_2}+\frac{1}{2}w(1)^2 b(1)M(1) -\frac{1}{2}w(0)^2 b(0)M(0).
\end{align}

Applying integration by parts and using the fact that $K_1(s,s)=K_2(s,s)$, we get
\begin{align}
\Gamma_3 = & \igzo w(s)^2 \left(a(s)\left[\pfs[K_1(s,t)-K_2(s,t)] \right]_{t=s} \right)ds \nonumber \\
&+\igzo \igzs w(s)a(s)K_{1,ss}(s,t)w(t)dtds \nonumber \\
&+\igzo \igso w(s)a(s)K_{2,ss}(s,t)w(t)dtds. \nonumber
\end{align} Using a change of order of integration as applied in Equation~\eqref{Gamma_3} in Lemma~\ref{lem:primal}, we obtain
\begin{align}
\Gamma_3 = & \igzo w(s)^2 \left(a(s)\left[\pfs[K_1(s,t)-K_2(s,t)] \right]_{t=s} \right)ds \nonumber \\
&+\igzo \igzs w(s)\left(\hlf a(s)K_{1,ss}(s,t)\right)w(t)dtds \nonumber \\
&+\igzo \igzs w(s)\left(\hlf a(t)K_{1,tt}(s,t)\right)w(t)dtds \nonumber \\
&+\igzo \igso w(s)\left(\hlf a(s)K_{2,ss}(s,t)\right)w(t)dtds \nonumber \\
&\label{eqn:dual_gamma_3}+\igzo \igso w(s)\left(\hlf a(t)K_{2,tt}(s,t)\right)w(t)dtds
\end{align} Similarly,
\begin{align}
&\Gamma_4 = \nonumber \\
 & \igzo \igzs w(s)\left(\hlf b(s)K_{1,s}(s,t)  +\hlf b(t)K_{1,t}(s,t)\right)w(t)dtds \nonumber \\
&\label{eqn:dual_gamma_4}+\igzo \igso w(s)\left(\hlf b(s)K_{2,s}(s,t) +\hlf b(t)K_{2,t}(s,t)\right)w(t)dtds
\end{align} and
\begin{align}
\Gamma_5 = &\igzo w(s)^2 M(s)c(s)ds \nonumber \\
& + \igzo \igzs w(s)\left(\hlf c(s)K_1(s,t)+\hlf c(t)K_1(s,t)\right)w(t)dtds \nonumber \\
&\label{eqn:dual_gamma_5} + \igzo \igso w(s)\left(\hlf c(s)K_2(s,t)+\hlf c(t)K_2(s,t)\right)w(t)dtds.
\end{align}

Substituting \eqref{eqn:dual_gamma_1}-\eqref{eqn:dual_gamma_5} in \eqref{eqn:dual_gamma},
\begin{align}
&\ip{\mathcal{A}\mathcal{P}w}{w}+\ip{w}{\mathcal{A}\mathcal{P}w} \nonumber \\
&\le  \ip{\bmat{w(1)\\w}}{ \mathcal{T} \bmat{w(1)\\w}}_{\R \times L_2}+w(0)\left(T_3 w(0) + T_4 w_s(0) \right) \nonumber \\
&\quad+ 2 \igzo w(1)a(1)K_{1,s}(1,s)w(s)ds+ 2w(1)a(1)M(1)w_s(1) \nonumber \\
&\label{eqn:proof_ineq}\quad+2w(1)a(1)M_s(1)w(1).
\end{align}

Since $w \in \mathcal{P}^{-1}\mathcal{D}_0$, there exists a $y \in \mathcal{D}_0$ such that $w=\mathcal{P}^{-1}y$ which implies $y=\mathcal{P}w$. Hence, we obtain the boundary condition
\[y_s(1)=M_s(1)w(1)+M(1)w_s(1)+\igzo K_{1,s}(1,s)w(s)ds.\]
Since $y \in D_0$, $y_s(1)=0$ and hence
\[
M_s(1)w_1(1)=-M(1)w_s(1)-\igzo K_{1,s}(1,s)w(s)ds.
\]
Substituting this boundary condition into the last term of \eqref{eqn:dual_setup} gives us the desired result.
\end{proof}

\begin{corollary}\label{cor:dual2}
Suppose we are given  $\{M,K_1,K_2\} \in \Xi_{d_1,d_2,\epsilon}$ and $\{T_0,T_1,T_2,T_3,T_4\}=\mathcal{N}_\epsilon(M,K_1,K_2)$. Then, for $\mathcal{A}$ as defined in Equation~\eqref{eqn:Aoperator} and $\mathcal{P}$ as defined in Equation~\eqref{eqn:Poperator}, we have that
\begin{align}
&\ip{\mathcal{A}\mathcal{P}w}{w}+\ip{w}{\mathcal{A}\mathcal{P}w} \\
&\le  \ip{\bmat{w(1)\\w}}{ \mathcal{T} \bmat{w(1)\\w}}_{\R \times L_2}+w(0)\left(T_3 w(0) + T_4 w_s(0) \right) \nonumber \\
&\quad+ 2 \igzo w(1)a(1)K_{1,s}(1,s)w(s)ds+ 2w(1)a(1)M(1)w_s(1) \nonumber \\
&\label{eqn:dual_setup}\quad+2w(1)a(1)M_s(1)w(1).
\end{align}
for any $w \in \mathcal{P}^{-1}\mathcal{D}$ where $\mathcal{D}$ is defined in Equation~\eqref{state_space_stab} and
\begin{align*}
(\mathcal{T}y)(s) :=  T_0(s)\bmat{y(1) \\ y(s)} &+   \igzs  \bmat{0 & 0 \\ 0 & T_1(s,t)} \bmat{y(1) \\ y(t)} dt \\
& + \igso  \bmat{0 & 0 \\ 0 & T_2(s,t)} \bmat{y(1) \\ y(t)} dt.
\end{align*}
\end{corollary}
The proof of Corollary~\ref{cor:dual2} is implied by the proof of Lemma~\ref{lem:dual} in Inequality~\eqref{eqn:proof_ineq}.

\subsection{Acknowledgements}

 This research was carried out with the financial support of the Chateaubriand fellowship program and NSF CAREER Grant CMMI-1151018.
 
\bibliographystyle{plain}
\bibliography{TAC}

\begin{thebibliography}{10}

\bibitem{baker1969lyapunov}
R.~Baker and A.~Bergen.
\newblock Lyapunov stability and {L}yapunov functions of infinite dimensional
  systems.
\newblock {\em IEEE Transactions on Automatic Control}, 14(4):325--334, 1969.

\bibitem{bensoussan1995representation}
A.~Bensoussan, G.D. Prato, M.~Delfour, S.~Mitter, and DL~Russell.
\newblock {Representation and control of infinite dimensional systems, vols. 1
  and 2}.
\newblock {\em SIAM Review}, 37(3):476--476, 1995.

\bibitem{blum1998complexity}
L.~Blum.
\newblock {\em {Complexity and real computation}}.
\newblock Springer Verlag, 1998.

\bibitem{bribiesca2011strict}
F.~Bribiesca~Argomedo, C.~Prieur, E.~Witrant, and S.~Br{\'e}mond.
\newblock A strict control {L}yapunov function for a diffusion equation with
  time-varying distributed coefficients.
\newblock {\em IEEE Transactions on Automatic Control}, 58(2):290--303, 2012.

\bibitem{byrnes1999example}
C.~I. Byrnes, D.~S. Gilliam, and V.~I. Shubov.
\newblock Example of output regulation for a system with unbounded inputs and
  outputs.
\newblock In {\em Decision and Control, 1999. Proceedings of the 38th IEEE
  Conference on}, volume~5, pages 4280--4284. IEEE, 1999.

\bibitem{chesi1999convexification}
G.~Chesi, A.~Tesi, A.~Vicino, and R.~Genesio.
\newblock On convexification of some minimum distance problems.
\newblock In {\em European control conference}, 1999.

\bibitem{coron2008dissipative}
J.~M. Coron, G.~Bastin, and B.~d'Andr{\'e}a Novel.
\newblock Dissipative boundary conditions for one-dimensional nonlinear
  hyperbolic systems.
\newblock {\em SIAM Journal on Control and Optimization}, 47(3):1460--1498,
  2008.

\bibitem{coron1998stabilization}
J.~M. Coron and B.~d'Andrea Novel.
\newblock Stabilization of a rotating body beam without damping.
\newblock {\em IEEE Transactions on Automatic Control}, 43(5):608--618, 1998.

\bibitem{coron2007strict}
J.~M. Coron, B.~d'Andrea Novel, and G.~Bastin.
\newblock A strict {L}yapunov function for boundary control of hyperbolic
  systems of conservation laws.
\newblock {\em IEEE Transactions on Automatic Control}, 52(1):2--11, 2007.

\bibitem{curtain1995introduction}
R.F. Curtain and H.J. Zwart.
\newblock {\em {An introduction to infinite-dimensional linear systems
  theory}}.
\newblock Springer, 1995.

\bibitem{datko1970extending}
R.~Datko.
\newblock Extending a theorem of {A}.{M}. {L}iapunov to {H}ilbert space.
\newblock {\em Journal of Mathematical analysis and applications},
  32(3):610--616, 1970.

\bibitem{egorov1996spectral}
Y.~Egorov and V.~Kondratiev.
\newblock {\em On spectral theory of elliptic operators, volume 89 of
  {O}perator Theory: {A}dvances and {A}pplications}.
\newblock Birkh{\"a}user Verlag, Basel, 1996.

\bibitem{fridman2009lmi}
E.~Fridman and Y.~Orlov.
\newblock An {LMI} approach to ${H}_\infty$ boundary control of semilinear
  parabolic and hyperbolic systems.
\newblock {\em Automatica}, 45(9):2060--2066, 2009.

\bibitem{hale1971functional}
J.~K. Hale.
\newblock {\em Functional differential equations}.
\newblock Springer, 1971.

\bibitem{hardy1952inequalities}
G.~H. Hardy, J.~E. Littlewood, and G.~Polya.
\newblock {\em Inequalities}.
\newblock Cambridge university press, 1952.

\bibitem{harkort2011discrete}
C.~Harkort and J.~Deutscher.
\newblock Discrete-time modal state reconstruction for infinite-dimensional
  systems using generalized sampling.
\newblock In {\em IFAC World Congress}, volume~18, pages 13311--13316, 2011.

\bibitem{henrion2003gloptipoly}
D.~Henrion and J.~B. Lasserre.
\newblock Gloptipoly: Global optimization over polynomials with {MATLAB} and
  {S}e{D}u{M}i.
\newblock {\em ACM Transactions on Mathematical Software (TOMS)},
  29(2):165--194, 2003.

\bibitem{jacobi2001representation}
T.~Jacobi.
\newblock A representation theorem for certain partially ordered commutative
  rings.
\newblock {\em Mathematische Zeitschrift}, 237(2):259--273, 2001.

\bibitem{krasovskistability}
N.~N. Krasovski.
\newblock Stability of motion. 1963.

\bibitem{krstic2008adaptive}
M.~Krstic and A.~Smyshlyaev.
\newblock Adaptive boundary control for unstable parabolic {PDE}s, {P}art {I}:
  Lyapunov design.
\newblock {\em IEEE Transactions on Automatic Control}, 53(7):1575--1591, 2008.

\bibitem{krstic2008boundary}
M.~Krstic and A.~Smyshlyaev.
\newblock {\em Boundary control of PDEs: A course on backstepping designs},
  volume~16.
\newblock Society for Industrial Mathematics, 2008.

\bibitem{lasiecka1980unified}
I.~Lasiecka.
\newblock Unified theory for abstract parabolic boundary problems, a semigroup
  approach.
\newblock {\em Applied Mathematics \& Optimization}, 6(1):287--333, 1980.

\bibitem{lasiecka1983feedback}
I.~Lasiecka and R.~Triggiani.
\newblock Feedback semigroups and cosine operators for boundary feedback
  parabolic and hyperbolic equations.
\newblock {\em Journal of Differential Equations}, 47(2):246--272, 1983.

\bibitem{lasiecka1994control}
I.~Lasiecka and R.~Triggiani.
\newblock Control and stabilization of distributed parameter systems;
  theoretical and computational aspects.
\newblock Technical report, DTIC Document, 1994.

\bibitem{lasiecka2000control}
I.~Lasiecka and R.~Triggiani.
\newblock {\em Control theory for partial differential equations: Volume 1,
  Abstract parabolic systems: Continuous and approximation theories}, volume~1.
\newblock Cambridge University Press, 2000.

\bibitem{lasserre2001global}
J.~B. Lasserre.
\newblock Global optimization with polynomials and the problem of moments.
\newblock {\em SIAM Journal on Optimization}, 11(3):796--817, 2001.

\bibitem{lions1972non}
J.~L. Lions and E.~Magenes.
\newblock {\em Non-Homogeneous Boundary Value Problems and Applications}.
\newblock Springer-Verlag, Berlin/New York, 1972.

\bibitem{murray2002mathematical}
J.~D. Murray.
\newblock {\em Mathematical biology}, volume~2.
\newblock Springer, 2002.

\bibitem{nagel2000one}
R.~Nagel.
\newblock {\em One-parameter semigroups for linear evolution equations}, volume
  194.
\newblock Springer, 2000.

\bibitem{papachristodoulou2005constructing}
A.~Papachristodoulou, M.~M. Peet, and S.~Lall.
\newblock Constructing {L}yapunov-{K}rasovskii functionals for linear time
  delay systems.
\newblock In {\em Proceedings of the American Control Conference}, pages
  2845--2850, 2005.

\bibitem{parrilo2000structured}
P.A. Parrilo.
\newblock {\em Structured semidefinite programs and semialgebraic geometry
  methods in robustness and optimization}.
\newblock PhD thesis, California Institute of Technology, 2000.

\bibitem{peet2013inverses}
M.~M. Peet.
\newblock Full-state feedback of delayed systems using {SOS}: A new theory of
  duality.
\newblock In {\em 11th IFAC Workshop on Time-Delay Systems}, 2013.

\bibitem{peetlmi}
M.~M. Peet.
\newblock {LMI} parametrization of {L}yapunov functions for
  infinite-dimensional systems: A toolbox.
\newblock {\em Submited to American Control Conference (ACC). Available on
  control.asu.edu}, 2014.

\bibitem{peet2008using}
M.~M. Peet and A.~Papachristodoulou.
\newblock Using polynomial semi-separable kernels to construct
  infinite-dimensional {L}yapunov functions.
\newblock In {\em 47th IEEE Conference on Decision and Control}, pages
  847--852, 2008.

\bibitem{peet2009inverses}
M.~M. Peet and Antonis Papachristodoulou.
\newblock Inverses of positive linear operators and state feedback design for
  timedelay systems.
\newblock In {\em 8th IFAC Workshop on Time-Delay Systems}, 2009.

\bibitem{powers1998algorithm}
V.~Powers and T.~W{\"o}rmann.
\newblock An algorithm for sums of squares of real polynomials.
\newblock {\em Journal of pure and applied algebra}, 127(1):99--104, 1998.

\bibitem{prajna2001introducing}
S.~Prajna, A.~Papachristodoulou, and P.~A. Parrilo.
\newblock Introducing {SOSTOOLS}: A general purpose sum of squares programming
  solver.
\newblock In {\em Proceedings of the 41st IEEE Conference on Decision and
  Control, 2002}, volume~1, pages 741--746, 2002.

\bibitem{putinar1993positive}
M.~Putinar.
\newblock Positive polynomials on compact semi-algebraic sets.
\newblock {\em Indiana University Mathematics Journal}, 42(3):969--984, 1993.

\bibitem{scheiderer2009positivity}
C.~Scheiderer.
\newblock Positivity and sums of squares: a guide to recent results.
\newblock In {\em Emerging applications of algebraic geometry}, pages 271--324.
  Springer, 2009.

\bibitem{schmudgen1991thek}
K.~Schm{\"u}dgen.
\newblock The {K}-moment problem for compact semi-algebraic sets.
\newblock {\em Mathematische Annalen}, 289(1):203--206, 1991.

\bibitem{shilov1974elementary}
G.~E. Shilov.
\newblock {\em Elementary functional analysis}.
\newblock Courier Dover Publications, 1974.

\bibitem{smyshlyaev2006lyapunov}
A.~Smyshlyaev.
\newblock Lyapunov adaptive boundary control for parabolic {PDE}s with
  spatially varying coefficients.
\newblock In {\em American Control Conference}, pages 41--48, 2006.

\bibitem{smyshlyaev2007adaptive}
A.~Smyshlyaev and M.~Krstic.
\newblock Adaptive boundary control for unstable parabolic {PDE}s, {P}art {II}:
  Estimation-based designs.
\newblock {\em Automatica}, 43(9):1543--1556, 2007.

\bibitem{smyshlyaev2007adaptive2}
A.~Smyshlyaev and M.~Krstic.
\newblock Adaptive boundary control for unstable parabolic {PDE}s, {P}art
  {III}: Output feedback examples with swapping identifiers.
\newblock {\em Automatica}, 43(9):1557--1564, 2007.

\bibitem{stengle1974nullstellensatz}
G.~Stengle.
\newblock A {N}ullstellensatz and a {P}ositivstellensatz in semialgebraic
  geometry.
\newblock {\em Mathematische Annalen}, 207(2):87--97, 1974.

\bibitem{tanaka2009sum}
K.~Tanaka, H.~Yoshida, H.~Ohtake, and H.O. Wang.
\newblock A sum-of-squares approach to modeling and control of nonlinear
  dynamical systems with polynomial fuzzy systems.
\newblock {\em IEEE Transactions on Fuzzy Systems}, 17(4):911--922, 2009.

\bibitem{triggiani1980well}
R.~Triggiani.
\newblock Well-posedness and regularity of boundary feedback parabolic systems.
\newblock {\em Journal of Differential Equations}, 36(3):347--362, 1980.

\bibitem{van1993h}
B.~Van~Keulen.
\newblock {\em ${H}_\infty$-control for distributed parameter systems: a state
  space approach}.
\newblock Birkhauser, 1993.

\bibitem{weiss1989representation}
G.~Weiss.
\newblock The representation of regular linear systems on hilbert spaces.
\newblock {\em International series of numerical mathematics}, 91:401--416,
  1989.

\bibitem{witrant2007control}
E.~Witrant, E.~Joffrin, S.~Br{\'e}mond, G.~Giruzzi, D.~Mazon, O.~Barana, and
  P.~Moreau.
\newblock A control-oriented model of the current profile in tokamak plasma.
\newblock {\em Plasma Physics and Controlled Fusion}, 49(7):1075, 2007.

\end{thebibliography}
\begin{IEEEbiographynophoto}{Aditya Gahlawat}
received the B.Tech degree in mechanical engineering from Punjabi University, Patiala, India in 2007, the M.S. degree in mechanical and aerospace engineering from Illinois Institute of Technology, Chicago, USA in 2009 and is currently pursuing a Ph.D. degree in mechanical and aerospace engineering from Illinois Institute of Technology, Chicago, USA and Universit\'{e} de Grenoble, St. Martin d'Heres, France. 

His research focuses on the application of convex optimization based methods for the analysis and control of systems governed by partial differential equations with application to thermonuclear fusion.

Aditya Gahlawat was awarded the Chateaubriand fellowship in 2011 and 2012.
\end{IEEEbiographynophoto}

\begin{IEEEbiographynophoto}{Matthew M. Peet}
received the B.S. degree in physics and in aerospace engineering from the University of Texas, Austin, TX, USA, in 1999 and the M.S. and Ph.D. degrees in aeronautics and astronautics from Stanford University, Stanford, CA, in 2001 and 2006,
respectively.

He was a Postdoctoral Fellow at the National Institute for Research in Computer Science and Control (INRIA), Paris, France, from 2006 to 2008, where he worked in the SISYPHE and BANG groups. He was an Assistant Professor of Aerospace Engineering in the Mechanical, Materials, and Aerospace Engineering Department, Illinois Institute of Technology, Chicago, IL, USA, from 2008 to 2012. Currently, he is an Assistant Professor of Aerospace Engineering, School for the Engineering of Matter, Transport, and Energy, Arizona State University, Tempe, AZ, USA, and Director of the Cybernetic Systems and Controls Laboratory. His research interests are in the role of computation as it is applied to the understanding and control of complex and large-scale systems. Applications include fusion energy and immunology.

Dr. Peet received a National Science Foundation CAREER award in 2011.
\end{IEEEbiographynophoto}
\end{document}